\documentclass[twocolumn]{autart}    % Enable this line and disable the 
\usepackage{amsmath,amsfonts}
\usepackage{algorithmic}
\usepackage{algorithm}
\usepackage{array}
\usepackage[caption=false,font=normalsize,labelfont=sf,textfont=sf]{subfig}
\usepackage{textcomp}
\usepackage{stfloats}
\usepackage{url}
\usepackage{verbatim}
\usepackage[numbers,sort&compress]{natbib}
\usepackage{bm}
\usepackage{amssymb} 
\usepackage{graphicx}          % Include this line if your 
\newtheorem{assumption}{Assumption}

\newtheorem{proposition}{Proposition}
\newtheorem{lemma}{Lemma}
\newtheorem{theorem}{Theorem}

\newtheorem{remark}{Remark}
\newtheorem{problem}{Problem}

\newcommand{\col}{\operatorname{col}}

\newcommand{\cone}{\operatorname{cone}}

\makeatletter
\newenvironment{proof}[1][Proof]{%
  \def\Elproofname{#1.}\pf
}{\endpf}
\def\@opargbegintheorem#1#2#3{%
  \trivlist
  \item[\hskip\labelsep{\bfseries #1\ #2.}]%
  {\normalfont\bfseries\raggedright (#3)\par}\nobreak\itshape\ignorespaces
}
\makeatother
\begin{document}
\sloppy

\begin{frontmatter}

\title{Reference-Governed Distributed Safe Gradient Flow for Safe Optimal Output Agreement of Multi-Agent Systems \thanksref{footnoteinfo}} % Title, preferably not more 

\thanks[footnoteinfo]{This paper was not presented at any IFAC 
meeting. Corresponding author Bo Yang. }

\author[China]{Zhanglin Shangguan}\ead{ditto331@sjtu.edu.cn},    % Add the 
\author[U]{Wei Xiao}\ead{weixy@mit.edu},
\author[China]{Bo Yang}\ead{bo.yang@sjtu.edu.cn},               % e-mail address 
\author[China]{Xinping Guan}\ead{xpguan@sjtu.edu.cn}.

\address[China]{Department of Automation and Intelligent Sensing, Shanghai Jiao Tong University, Shanghai, China}  
\address[U]{Computer Science and Artificial Intelligence Laboratory, Massachusetts Institute of Technology, Cambridge, USA}
% Please supply                                              
% \address[Rome]{Senate House, Rome}             % full addresses
% \address[Baiae]{The White House, Baiae}        % here.

                             % keyword list or with the 
                                          % help of the Automatica 
                                          % keyword wizard

\begin{abstract}
This paper studies safe optimal output agreement for nonlinear multi-agent systems with output safety constraints. Existing safe feedback optimization methods often implement gradient-flow dynamics directly through the plant input, which may require high-order control barrier functions (HOCBFs). The resulting derivative-chain design is tuning-sensitive and can introduce additional equilibrium conditions that alter the steady-state optimal solution. We propose a reference-governed two-layer architecture that separates lower-layer output regulation from upper-layer distributed optimization. The upper layer filters the reference gradient flow through first-order control barrier function constraints, which are easier to tune and preserve the steady-state optimality structure of the original agreement problem. The lower layer uses an internal-model-based output regulator with a reference-dependent Lyapunov function, from which dynamic safety margins (DSMs) are constructed to certify transient output safety. We prove forward invariance, optimal-solution preservation under DSM-compatibility conditions, and convergence via a Lyapunov small-gain argument. Simulations validate safe convergence, show advantages over HOCBF-based feedback optimization, and demonstrate adaptive tangential objective shaping for escaping spurious equilibria induced by nonconvex obstacles.
\end{abstract}

\end{frontmatter}

\section{Introduction}

Feedback optimization has emerged as a control-oriented approach for real-time
optimization of dynamical systems, closely related to extremum seeking but
explicitly accounting for plant dynamics and closed-loop stability. Instead of
solving a static optimization problem offline, feedback optimization uses
measured outputs to steer the closed-loop steady state toward an optimal
operating point \cite{colombino2019online,liu2025feedback,rodrigues2025event}.
For networked systems, this idea leads to distributed optimal agreement, where
agents exchange local information to agree on an output value minimizing an
aggregate objective \cite{xie2019global}. Such formulations arise in
multi-robot coordination, power networks, transportation systems, and other
cyber-physical applications
\cite{pichierri2026multi,bernstein2019real,li2026distributed}. However,
optimality and stability alone are insufficient for safety-critical systems:
even when a feedback optimizer can steer the steady-state outputs to the optimal solution of the constrained steady-state problem, the physical outputs may violate
safety constraints during transients.

A common design route uses the plant's input-to-steady-state map to implement gradient or primal-dual dynamics together with consensus terms, thereby steering the steady-state outputs to an optimal agreement point
\cite{nedic2018distributed,carnevale2024nonconvex}. Inequality constraints are
often handled by projected dual or projected primal-dual dynamics
\cite{wang2024robust}, but these projections do not directly certify transient
output safety. Control barrier functions (CBFs) enforce forward invariance of safe sets, while high-order CBFs (HOCBFs) extend this idea to safety constraints with high relative degree \cite{ames2016control,xiao2021high}. CBF-filtered gradient flows have been developed for constrained optimization
\cite{allibhoy2023control} and recently embedded into feedback optimization
through an input-flow formulation, where the plant input is dynamically updated
according to constrained optimization dynamics
\cite{delimpaltadakis2024continuous,delimpaltadakis2025feedback}. However,
because safety is enforced through the input flow, output safety constraints
generally become high-relative-degree conditions and require HOCBFs. The
resulting safety filters may introduce additional stationarity conditions that
are not part of the original steady-state optimization problem, thereby
preventing convergence to the true optimum.

Another line of work separates the optimization layer from the physical
tracking layer. The upper layer generates a time-varying reference, while the
lower layer tracks this reference using a stabilizing controller. This
structure enables small-gain analysis of the optimizer--plant interconnection \cite{jiang1994small,dashkovskiy2010small},
where tracking errors perturb the optimization dynamics and reference
velocities perturb the tracker \cite{liu2021distributed}. Inequality-constrained
extensions have been studied using projected primal-dual reference dynamics
\cite{qin2023distributed}. More recently, expanded and contracted safety
constraints have been constructed separately for the optimizer and controller
layers to handle optimal solutions located on safety boundaries \cite{ma2025distributed}, but
this approach is primarily tailored to convex safety constraints. Related
reference-governor methods drive a pre-stabilized system by an auxiliary
reference that moves toward a target while preserving constraints
\cite{nicotra2018explicit,garone2015explicit}. Dynamic safety margins (DSMs)
further quantify, through a reference-dependent Lyapunov function, the
remaining safety margin during transient tracking and can be interpreted as
CBFs for an augmented state-reference system
\cite{freire2026using,nakano2026optimization}. These ideas suggest enforcing safety at the reference-dynamics level rather
than through HOCBFs in the input flow. Nevertheless, if the reference dynamics are implemented as safe gradient flows filtered by CBF-based quadratic programs (CBF-QPs), the upper-layer optimizer may still suffer from undesired equilibria caused by active safety constraints. This issue becomes particularly pronounced for nonconvex safe sets, where the safety correction imposed by active CBF constraints can counteract the nominal descent direction along the boundary of the safe reference set, producing spurious boundary equilibria, including asymptotically stable ones \cite{tan2024undesired,REIS2026113032}.

This paper develops a reference-governed distributed safe gradient-flow
framework for safe optimal output agreement. The goal is to guarantee strict
transient output safety while preserving the optimality of the original static
agreement problem. We also examine nonconvex obstacle constraints, where
CBF-filtered flows may admit spurious stable boundary equilibria. Building on the gradient-similarity-based tangential excitation design developed in
\cite{shangguan2026synthesizing}, we introduce an adaptive tangential objective
shaping mechanism that locally changes the boundary geometry of the reference
gradient flow and can turn attracting nonconvex-induced equilibria into
saddle-type ones without relaxing safety constraints.

The main contributions are summarized as follows:
\begin{itemize}
    \item We propose a two-layer architecture that separates
    internal-model-based output regulation from distributed safe optimization.
    Using a reference-dependent tracking Lyapunov function, we construct DSMs
    that certify transient output safety during reference motion.

    \item We design a distributed reference governor that filters the nominal
    gradient-consensus flow through first-order CBF constraints on the
    reference dynamics. This avoids HOCBF constructions in the plant-input
    channel and preserves the KKT geometry of the static optimal agreement
    problem under suitable DSM-compatibility conditions.

    \item We prove convergence of the coupled tracking-governor dynamics by a
    QP perturbation estimate and a small-gain argument. Simulations validate
    safe convergence, compare against HOCBF-based and projected primal-dual
    feedback optimization baselines, and demonstrate that adaptive tangential
    objective shaping can help escape spurious equilibria induced by
    nonconvex obstacles.
\end{itemize}

\textbf{Notation.}
For vectors $z_1,\ldots,z_m$, $\operatorname{col}(z_1,\ldots,z_m)$ denotes their vertical concatenation. For a matrix $A$, $\ker A:=\{x:Ax=0\}$ denotes its null space, and $A\otimes B$ denotes the Kronecker product. The symbol $I_p$ denotes the $p$-dimensional identity matrix, and $1_N$ denotes the $N$-dimensional vector of all ones. For vectors $v_1,\ldots,v_m$, $\operatorname{span}\{v_1,\ldots,v_m\}$ denotes their linear span, and $\operatorname{cone}\{v_1,\ldots,v_m\}:={\sum_{\ell=1}^m\alpha_\ell v_\ell:\alpha_\ell\ge0}$ denotes their conic hull. 

\section{Reference-Governed Safe Output Agreement}
\label{sec:problem}

Consider a network of $N$ nonlinear agents indexed by
$\mathcal N:=\{1,\ldots,N\}$. Agent $i$ is described by
\begin{equation}
    \dot x_i=f_i(x_i,u_i),\qquad y_i=\ell_i(x_i),
    \label{eq:agent}
\end{equation}
where $x_i\in\mathbb R^{n_i}$, $u_i\in\mathbb R^{m_i}$,
and $y_i\in\mathbb R^p$ are the state, input, and output,
respectively.

The output safety requirement of agent $i$ is
\begin{equation}
    \eta_i(y_i)\ge0,\qquad i\in\mathcal N,
    \label{eq:output-safety}
\end{equation}
where $\eta_i(y_i):=\col\big(\eta_{ik}(y_i)\big)_{k\in\mathcal K_i}\in\mathbb R^{q_i}.$

All vector inequalities are understood componentwise. The constraints
\eqref{eq:output-safety} may encode workspace limits, output-level
collision-avoidance constraints after responsibility allocation, or other
task-dependent safety specifications.

The agents exchange information over a weighted directed graph
$\mathcal G=(\mathcal N,\mathcal E,A)$ with adjacency matrix
$A=[a_{ij}]$. Its Laplacian is $L=[l_{ij}]$, where $ l_{ii}=\sum_{j\in\mathcal N}a_{ij}, l_{ij}=-a_{ij}, i\ne j.$ Let $L_\otimes:=L\otimes I_p$.

\begin{assumption}[Communication graph]
\label{ass:graph}
The graph $\mathcal G$ is strongly connected and weight-balanced. Hence $\ker L=\ker L^\top=\operatorname{span}\{1_N\}.$
\end{assumption}

Each agent has a local objective function
$c_i:\mathbb R^p\to\mathbb R$, which evaluates the cost associated with an assigned output value.

\begin{assumption}[Local objectives and output feasibility]
\label{ass:cost}
For each $i\in\mathcal N$, the functions $c_i$ and
$\eta_{ik}$, $k\in\mathcal K_i$, are continuously differentiable, and there
exist $\omega_i,\vartheta_i>0$ such that, for all $r_1,r_2\in\mathbb R^p$,
$\big(\nabla c_i(r_1)-\nabla c_i(r_2)\big)^\top(r_1-r_2)\ge
\omega_i\|r_1-r_2\|^2$ and
$\|\nabla c_i(r_1)-\nabla c_i(r_2)\|\le\vartheta_i\|r_1-r_2\|$.
The common output-feasible set
$\mathcal S:=\{r\in\mathbb R^p:\eta_i(r)\ge0,\ i\in\mathcal N\}$ is
nonempty; when global optimality is claimed, $\mathcal S$ is assumed convex.
\end{assumption}

Before considering the agent dynamics, the desired agreement value is
defined by the following static distributed optimization problem:
\begin{equation}
\begin{aligned}
    \min_{r_1,\ldots,r_N\in\mathbb R^p}\quad&
        \sum_{i=1}^N c_i(r_i)\\
    \mathrm{s.t.}\quad&
        r_i=r_j,\qquad i,j\in\mathcal N,\\
    &
        \eta_i(r_i)\ge0,\qquad i\in\mathcal N .
\end{aligned}
\label{eq:static-distributed-optimization}
\end{equation}
Equivalently, because all feasible solutions satisfy the agreement constraint, \eqref{eq:static-distributed-optimization} reduces to $\min_{r\in\mathbb R^p} C(r):=\sum_{i=1}^N c_i(r), \mathrm{s.t.} r\in\mathcal S.$

Under Assumption~\ref{ass:cost}, $C$ is strongly convex with parameter
$\sum_{i=1}^N\omega_i$ and has Lipschitz continuous gradient with constant
$\sum_{i=1}^N\vartheta_i$. Hence, if $\mathcal S$ is convex,
\eqref{eq:static-distributed-optimization} has a unique global optimal solution, denoted by $r^\star$. 

When the output variables are generated by the nonlinear agents
\eqref{eq:agent}, the static problem becomes a control problem: make the
physical outputs converge to the optimal solution of
\eqref{eq:static-distributed-optimization} while satisfying the output safety
constraints along the entire trajectory.

\begin{problem}[Safe optimal output agreement]
\label{prob:safe-optimal-output-agreement}
Design distributed controllers $u_i$, using only locally available
information from neighboring agents, such that, for every admissible initial
condition, $\eta_i(y_i(t))\ge0$ for all $t\ge0$ and
$\lim_{t\to\infty}y_i(t)=y_i^\star=r^\star$, $i\in\mathcal N$, where
$r^\star$ denotes the optimal solution of \eqref{eq:static-distributed-optimization}.
\end{problem}

Problem~\ref{prob:safe-optimal-output-agreement} couples optimization,
agreement, tracking, and safety at the plant-input level. To solve it in a
modular way, we introduce a two-layer architecture. The lower layer is a
tracking controller that makes each physical output follow a commanded
reference, while the upper layer is a distributed reference governor that
generates these references and drives them toward the optimal solution $r^\star$.

For each agent $i\in\mathcal N$, let $g_i\in\mathbb R^p$ denote the reference
supplied by the reference governor to the lower tracking layer. The lower tracking layer may be dynamic, as is common in internal-model-based nonlinear output regulation~\cite{huang2004nonlinear}. We therefore introduce the augmented tracking-layer state $\chi_i:=\col(x_i,\zeta_i)$, where $\zeta_i$ collects the states of the dynamic tracking controller. For instance, in a PI tracking controller, $\zeta_i$ can be the integral state of the tracking error; in an internal-model regulator, it represents the compensator state used to generate the required steady-state input. This construction will be specified in detail in Section~\ref{sec:dsm}; if the tracking controller is static, then $\zeta_i$ is absent and $\chi_i=x_i$. The input applied to the plant is generated by the lower tracking layer and is denoted by $u_i=u_i(\chi_i,g_i).$ Accordingly, the closed-loop tracking dynamics of agent $i$ are written in the
compact form
\begin{equation}
    \dot\chi_i
    =
    \mathcal F_i\big(\chi_i,u_i(\chi_i,g_i)\big),
    \qquad
    y_i=h_i(\chi_i)=\ell_i(x_i).
    \label{eq:agent-tracking-closed-loop}
\end{equation}
The detailed realization of the tracking layer will be specified in
Section~\ref{sec:dsm}; here, \eqref{eq:agent-tracking-closed-loop} is only used
as a compact problem-level representation.

Following the reference-governor viewpoint, we regard $\rho_i:=\dot g_i$ as the virtual input of the reference dynamics. Thus, for each
$i\in\mathcal N$, the augmented system seen by the reference governor is
\begin{equation}
    \begin{pmatrix}
        \dot\chi_i\\
        \dot g_i
    \end{pmatrix}
    =
    \begin{pmatrix}
        \mathcal F_i\big(\chi_i,u_i(\chi_i,g_i)\big)\\
        \rho_i
    \end{pmatrix},
    \quad
    y_i=h_i(\chi_i).
    \label{eq:augmented-agent}
\end{equation}

For a fixed reference $g_i\equiv \bar g_i$, let
$\chi_i^\star(\bar g_i)$ denote the corresponding nominal tracking-layer
equilibrium. At this equilibrium, the lower layer is required to achieve output
matching, namely
\begin{equation}
    \mathcal F_i\big(
        \chi_i^\star(\bar g_i),
        u_i(\chi_i^\star(\bar g_i),\bar g_i)
    \big)=0,
    \quad
    h_i\big(\chi_i^\star(\bar g_i)\big)=\bar g_i .
    \label{eq:agent-frozen-reference-equilibrium}
\end{equation}
Therefore, the steady-state output safety constraint induced by
$\eta_{ik}(y_i)\ge0$ can be expressed in the reference space as
$\eta_{ik}(\bar g_i)\ge0$. This motivates the steady-state admissible
reference set
\begin{equation}
    \mathcal G_i
    :=
    \left\{
        g_i\in\mathbb R^p
        \mid
        \eta_{ik}(g_i)\ge0,\ 
        k\in\mathcal K_i
    \right\}.
    \label{eq:steady-state-admissible-reference-set}
\end{equation}

The reference governor is designed so that
$\lim_{t\to\infty}g_i(t)=r^\star$ while keeping
$g_i(t)\in\mathcal G_i$ for all $t\ge0$.

The proposed reference-governed two-layer architecture is illustrated in
Fig.~\ref{fig:sgf_framework}.
\begin{figure}[t]
    \centering
    \includegraphics[width=0.9\linewidth]{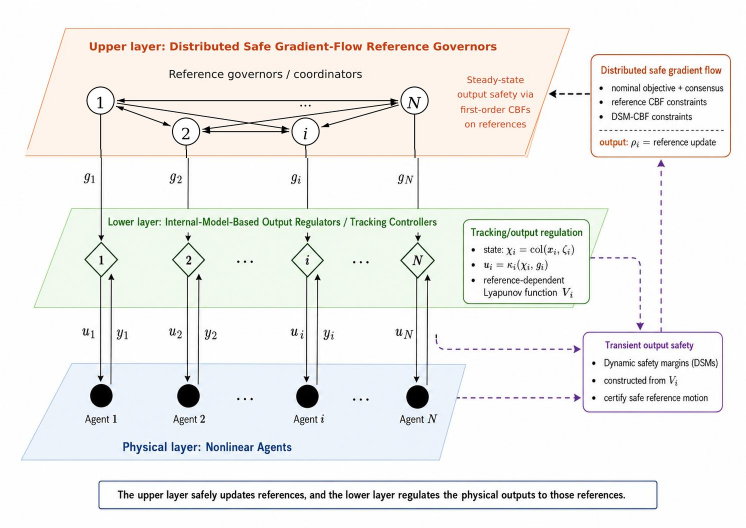}
    \caption{Reference-governed distributed safe-gradient-flow framework.}
    \label{fig:sgf_framework}
\end{figure}

However, steady-state admissibility of $g_i(t)$ does not by itself guarantee
transient output safety: during tracking,
$y_i(t)=h_i(\chi_i(t))\neq g_i(t)$ in general, so
$\eta_{ik}(g_i(t))\ge0$ does not imply $\eta_{ik}(y_i(t))\ge0$. Thus,
transient safety must be enforced by regulating the reference velocity
$\rho_i$.

Let $\mathcal A_0$ denote the admissible set of initial conditions for the
augmented system \eqref{eq:augmented-agent}. The upper-layer reference governor
designs distributed dynamics
\begin{equation}
    \dot g_i=\rho_i,\qquad i\in\mathcal N,
    \label{eq:reference-governor}
\end{equation}
where $\rho_i$ uses only local and neighboring information, such that, for
every $(\chi(0),g(0))\in\mathcal A_0$, the transient safety constraints
$\eta_{ik}(y_i(t))\ge0$, the reference constraints $g_i(t)\in\mathcal G_i$,
and the convergence limits $g_i(t),y_i(t)\to r^\star$ all hold.

\section{Transient Margins in the Tracking Layer}
\label{sec:dsm}

This section constructs the transient margins used by the reference governor.
A dynamic safety margin (DSM) measures, through a Lyapunov function, the
remaining transient energy before a safety boundary can be reached. Here the
energy is measured from the nominal output-regulation state associated with a
frozen reference $g_i$, not from the origin. The DSMs are therefore transient
tracking certificates, whereas the constraints in
\eqref{eq:steady-state-admissible-reference-set} are steady-state reference
admissibility constraints.

\subsection{Internal-model tracking layer}
\label{subsec:internal-model-tracking-layer}

For a fixed reference $g_i$, exact output regulation of agent $i$ requires a
plant state $\pi_i(g_i)$ and a steady-state input $u_i^{\rm ss}(g_i)$ satisfying
the local regulator equations
\begin{equation}
    0=f_i(\pi_i(g_i),u_i^{\rm ss}(g_i)),
    \qquad
    \ell_i(\pi_i(g_i))=g_i .
    \label{eq:agent-frozen-regulator-equations}
\end{equation}
The internal model dynamically generates this steady-state compensation when it is not explicitly known or robustly implementable. 

Accordingly, the tracking controller of agent $i$ contains an internal-model state $\zeta_i$ and is written in the general form \cite{1369396}
\begin{equation}
    \dot\zeta_i=\varphi_i(x_i,\zeta_i,g_i),
    \qquad
    u_i=\kappa_i(x_i,\zeta_i,g_i).
    \label{eq:agent-internal-model-controller}
\end{equation}
The state $\zeta_i$ is a controller state used to reproduce the compensating
signal implied by the regulator equations.

Define $\chi_i:=\col(x_i,\zeta_i)$ and $h_i(\chi_i):=\ell_i(x_i)$. For a frozen
reference $g_i$, the local augmented closed-loop tracking dynamics induced by
\eqref{eq:agent-internal-model-controller} are denoted compactly by
\begin{equation}
    \dot\chi_i=\mathcal F_i^{\rm cl}(\chi_i,g_i).
    \label{eq:agent-augmented-frozen-tracking-dynamics}
\end{equation}

The controller is assumed to be designed so that, for every constant
$g_i\in\mathcal G_i$, there exists a nominal output-regulation state
$\chi_i^\star(g_i)=\col(\pi_i(g_i),\theta_i(g_i))$ satisfying
\begin{equation}
    \mathcal F_i^{\rm cl}(\chi_i^\star(g_i),g_i)=0,
    \qquad
    h_i(\chi_i^\star(g_i))=g_i .
    \label{eq:agent-closed-loop-regulator-equations}
\end{equation}
Here $\theta_i(g_i)$ is the internal-model state realizing
$u_i^{\rm ss}(g_i)=\kappa_i(\pi_i(g_i),\theta_i(g_i),g_i)$. The tracking
coordinate is $\tilde\chi_i:=\chi_i-\chi_i^\star(g_i)$; for frozen $g_i$, the
tracking objective $\tilde\chi_i(t)\to0$ implies $h_i(\chi_i(t))\to g_i$.

\subsection{Internal-model tracking Lyapunov functions}
\label{subsec:internal-model-tracking-lyapunov}

For every $g_i\in\mathcal G_i$, assume that the frozen-reference system
\eqref{eq:agent-augmented-frozen-tracking-dynamics} admits a continuously
differentiable Lyapunov function
$V_i:\mathbb R^{\dim\chi_i}\times\mathcal G_i\to\mathbb R_{\ge0}$ measuring
the distance from $\chi_i$ to $\chi_i^\star(g_i)$.

More precisely, assume that for each $g_i\in\mathcal G_i$ there exists a domain $\mathcal D_i(g_i)\subseteq\mathbb R^{\dim\chi_i}$ containing $\chi_i^\star(g_i)$ and class-$\mathcal K$ functions $\alpha_{i,1},\alpha_{i,2}$ such that, for all $\chi_i\in\mathcal D_i(g_i)$,
\begin{equation}
    \alpha_{i,1}\!\left(\|\chi_i-\chi_i^\star(g_i)\|\right)
    \le
    V_i(\chi_i,g_i)
    \le
    \alpha_{i,2}\!\left(\|\chi_i-\chi_i^\star(g_i)\|\right).
    \label{eq:agent-tracking-lyapunov-bounds}
\end{equation}
Moreover, along the frozen-reference tracking dynamics,
\begin{equation}
    \frac{\partial V_i}{\partial\chi_i}(\chi_i,g_i)
    \mathcal F_i^{\rm cl}(\chi_i,g_i)
    \le
    -W_i(\chi_i,g_i)
    \le0,
    \label{eq:agent-tracking-lyapunov-decay}
\end{equation}
where $W_i$ is continuous and nonnegative. The largest invariant subset of $\{\chi_i\in\mathcal D_i(g_i): W_i(\chi_i,g_i)=0\}$ under $\dot\chi_i=\mathcal F_i^{\rm cl}(\chi_i,g_i)$ is the singleton $\{\chi_i^\star(g_i)\}$. Hence, if $g_i$ is held constant and $\chi_i(0)\in\mathcal D_i(g_i)$, then $\chi_i(t)\to\chi_i^\star(g_i)$ and $h_i(\chi_i(t))\to g_i$.

This Lyapunov function is the energy measure used below. Since
$V_i(\cdot,g_i)$ is nonincreasing for frozen $g_i$, its sublevel sets inside
$\mathcal D_i(g_i)$ are forward invariant.

\subsection{Safety DSM}
\label{subsec:safety-dsm}

Let the $k$-th output safety constraint of agent $i$ be $\eta_{ik}(y_i)\ge0$,
with safe output set $\mathcal S_{ik}^{y}:=\{y_i:\eta_{ik}(y_i)\ge0\}$ and
boundary $\partial\mathcal S_{ik}^{y}:=\{y_i:\eta_{ik}(y_i)=0\}$. For a frozen
reference $g_i$, lift this boundary to the augmented tracking-state space as
\begin{equation}
    \partial\mathcal X_{ik}^{\mathsf S}(g_i)
    :=
    \left\{
    \chi_i\in\mathcal D_i(g_i):
    \eta_{ik}\big(h_i(\chi_i)\big)=0
    \right\}.
    \label{eq:agent-lifted-output-safety-boundary}
\end{equation}
The associated safety energy threshold is
\begin{equation}
    \Gamma_{ik}^{\mathsf S}(g_i)
    :=
    \inf_{\chi_i\in\partial\mathcal X_{ik}^{\mathsf S}(g_i)}
    V_i(\chi_i,g_i),
    \label{eq:agent-safety-energy-threshold}
\end{equation}
with the convention $\Gamma_{ik}^{\mathsf S}(g_i)=+\infty$ if the lifted
boundary is empty. Equivalently, this minimizes $V_i$ over all augmented states
that can realize boundary outputs, so no one-to-one output map is required.

The safety DSM associated with the constraint $\eta_{ik}(y_i)\ge0$ is
\begin{equation}
    \mathsf M_{ik}^{\mathsf S}(\chi_i,g_i)
    :=
    \Gamma_{ik}^{\mathsf S}(g_i)-V_i(\chi_i,g_i).
    \label{eq:agent-safety-dsm}
\end{equation}
If $\mathsf M_{ik}^{\mathsf S}(\chi_i,g_i)>0$, the sublevel set
$\Omega_i(\chi_i,g_i):=\{\xi_i\in\mathcal D_i(g_i):V_i(\xi_i,g_i)\le
V_i(\chi_i,g_i)\}$ does not intersect
$\partial\mathcal X_{ik}^{\mathsf S}(g_i)$. Since this sublevel set is forward
invariant for frozen $g_i$, the transient output cannot reach
$\eta_{ik}(y_i)=0$ before tracking converges.

\subsection{Stability-domain DSM}
\label{subsec:stability-domain-dsm}

When $V_i(\cdot,g_i)$ is valid only on a restricted domain $\mathcal D_i(g_i)$,
the augmented state must also remain inside that Lyapunov domain. Define
\begin{equation}
    \Gamma_i^{\mathsf D}(g_i)
    :=
    \inf_{\chi_i\in\partial\mathcal D_i(g_i)}
    V_i(\chi_i,g_i),
    \label{eq:agent-domain-energy-threshold}
\end{equation}
provided $\mathcal D_i(g_i)\ne\mathbb R^{\dim\chi_i}$; for global Lyapunov
functions this margin is omitted. Given $\varepsilon_{iD}\in(0,1)$, define
\begin{equation}
    \mathsf M_i^{\mathsf D}(\chi_i,g_i)
    :=
    (1-\varepsilon_{iD})\Gamma_i^{\mathsf D}(g_i)-V_i(\chi_i,g_i).
    \label{eq:agent-stability-domain-dsm}
\end{equation}
The factor $1-\varepsilon_{iD}$ keeps the admissible sublevel set strictly
inside $\mathcal D_i(g_i)$. Thus $\mathsf M_i^{\mathsf D}\ge0$ certifies
Lyapunov-domain validity, while $\mathsf M_{ik}^{\mathsf S}\ge0$ certifies
output safety; both margins only regulate transient reference motion.

\section{Distributed Safe Gradient Flow}
\label{sec:dsgf}

This section designs the upper-layer reference governor
\eqref{eq:reference-governor}. It starts from the nominal output-gradient
direction $-\nabla c_i(y_i)$ and filters $\rho_i$ through first-order CBF
conditions $\dot b+\alpha_b b\ge0$, which make $\{b\ge0\}$ forward invariant.
Here $b$ is chosen as $\eta_{ik}(g_i)$,
$\mathsf M_{ik}^{\mathsf S}(\chi_i,g_i)$, or
$\mathsf M_i^{\mathsf D}(\chi_i,g_i)$ to enforce reference admissibility,
transient output safety, and Lyapunov-domain validity.

First, reference agreement is imposed by the equality-flow constraint
\begin{equation}
    \sum_{j\in\mathcal N}a_{ji}(\rho_i-\rho_j)
    +
    \alpha_c\sum_{j\in\mathcal N}a_{ji}(g_i-g_j)=0,
    \qquad i\in\mathcal N,
    \label{eq:consensus-flow-local}
\end{equation}
where $\alpha_c>0$. Equivalently,
$L_\otimes^\top\rho+\alpha_cL_\otimes^\top g=0$, so the consensus residual
decays exponentially.

Second, applying the CBF condition to $\eta_{ik}(g_i)$ gives
\begin{equation}
    \nabla\eta_{ik}(g_i)^\top\rho_i
    +
    \alpha_{\eta}\eta_{ik}(g_i)
    \ge0,
    \qquad
    \alpha_{\eta}>0 .
    \label{eq:ss-cbf-local}
\end{equation}
If $g_i(0)\in\mathcal G_i$, then $g_i(t)\in\mathcal G_i$ for all $t\ge0$.

Third, applying the CBF condition to the safety DSMs gives
\begin{equation}
\begin{aligned}
    &
    \nabla_{\chi_i}\mathsf M_{ik}^{\mathsf S}(\chi_i,g_i)^\top \dot\chi_i
    +
    \nabla_{g_i}\mathsf M_{ik}^{\mathsf S}(\chi_i,g_i)^\top \rho_i  \\
    &\hspace{3em}
    +
    \alpha_{\mathsf S}
    \mathsf M_{ik}^{\mathsf S}(\chi_i,g_i)
    \ge0,
    \qquad
    i\in\mathcal N,\ k\in\mathcal K_i,
\end{aligned}
\label{eq:safety-dsm-cbf-local}
\end{equation}
where $\alpha_{\mathsf S}>0$ and $\dot\chi_i$ is the first component of
\eqref{eq:augmented-agent}.

Finally, let $\mathcal N_{\mathsf D}\subseteq\mathcal N$ collect the agents
requiring a stability-domain DSM. For each $i\in\mathcal N_{\mathsf D}$,
\begin{equation}
\begin{aligned}
    &
    \nabla_{\chi_i}\mathsf M_i^{\mathsf D}(\chi_i,g_i)^\top \dot\chi_i
    +
    \nabla_{g_i}\mathsf M_i^{\mathsf D}(\chi_i,g_i)^\top \rho_i  \\
    &\hspace{3em}
    +
    \alpha_{\mathsf D}
    \mathsf M_i^{\mathsf D}(\chi_i,g_i)
    \ge0,
\end{aligned}
\label{eq:stability-dsm-cbf-local}
\end{equation}
where $\alpha_{\mathsf D}>0$; this constraint is omitted for agents with global
tracking Lyapunov functions.

The distributed safe gradient-flow reference governor is defined as the
solution of the following network-sparse QP:
\begingroup
\small
\setlength{\jot}{2pt}
\setlength{\arraycolsep}{2pt}
\begin{equation}
\begin{aligned}
    \rho^\star
    &=
    \arg\min_{\rho_1,\ldots,\rho_N}\quad
        \frac12\sum_{i=1}^N
        \left\|\rho_i+\nabla c_i(y_i)\right\|^2                                      \\[0.2em]
    \mathrm{s.t.}\quad
    &
    \sum_{j\in\mathcal N}a_{ji}(\rho_i-\rho_j)
    +
    \alpha_c\sum_{j\in\mathcal N}a_{ji}(g_i-g_j)=0,
    \quad i\in\mathcal N,                                                          \\[0.2em]
    &
    \nabla\eta_{ik}(g_i)^\top\rho_i
    +
    \alpha_{\eta}\eta_{ik}(g_i)
    \ge0,
    \quad i\in\mathcal N,\ k\in\mathcal K_i,                                       \\[0.2em]
    &
    \nabla_{\chi_i}\mathsf M_{ik}^{\mathsf S}(\chi_i,g_i)^\top \dot\chi_i
    +
    \nabla_{g_i}\mathsf M_{ik}^{\mathsf S}(\chi_i,g_i)^\top \rho_i                 \\
    &\hspace{4em}
    +
    \alpha_{\mathsf S}
    \mathsf M_{ik}^{\mathsf S}(\chi_i,g_i)
    \ge0,
    \quad i\in\mathcal N,\ k\in\mathcal K_i,                                      \\[0.2em]
    &
    \nabla_{\chi_i}\mathsf M_i^{\mathsf D}(\chi_i,g_i)^\top \dot\chi_i
    +
    \nabla_{g_i}\mathsf M_i^{\mathsf D}(\chi_i,g_i)^\top \rho_i                   \\
    &\hspace{4em}
    +
    \alpha_{\mathsf D}
    \mathsf M_i^{\mathsf D}(\chi_i,g_i)
    \ge0,
    \quad i\in\mathcal N_{\mathsf D}.
\end{aligned}
\label{eq:dsgf-qp-local}
\end{equation}
\endgroup
The upper-layer reference governor is then
\begin{equation}
    \dot g_i=\rho_i^\star,
    \qquad i\in\mathcal N .
    \label{eq:dsgf-reference-governor}
\end{equation}

The objective keeps $\rho_i$ close to $-\nabla c_i(y_i)$, while the constraints
preserve agreement, steady-state admissibility, transient output safety, and
Lyapunov-domain validity. All objective and CBF terms are local; only
\eqref{eq:consensus-flow-local} couples neighboring velocities, so the QP has
a sparse network structure suitable for distributed QP solvers.

By the CBF comparison argument, if the QP remains feasible and initially
$
    g_i(0)\in\mathcal G_i,
    \mathsf M_{ik}^{\mathsf S}(\chi_i(0),g_i(0))\ge0,
    \mathsf M_i^{\mathsf D}(\chi_i(0),g_i(0))\ge0
$
for all relevant $i$ and $k$, then these inequalities remain true for all
$t\ge0$.

% \begin{assumption}[CBF-QP feasibility and regularity]
% \label{ass:qp}
% There exists a compact set containing the closed-loop trajectory on which
% \eqref{eq:dsgf-qp-local} is feasible, the CBF data are continuously
% differentiable, a standard constraint qualification holds, and
% $(\chi,g)\mapsto\rho^\star(\chi,g)$ is locally Lipschitz.
% \end{assumption}

\begin{remark}[Trivial feasibility of the CBF constraints]
The CBF inequalities in \eqref{eq:dsgf-qp-local} inherit the
trivial-update feasibility property of DSM-based reference governors:
whenever $g_i\in\mathcal G_i$, $M^S_{ik}(\chi_i,g_i)\ge0$, and
$M^D_i(\chi_i,g_i)\ge0$, the frozen-reference choice $\rho_i=0$
satisfies the reference-CBF and DSM-CBF inequalities. This follows from
the forward invariance of the Lyapunov sublevel sets for fixed references,
as in the explicit reference governor-CBF feasibility argument of~\cite{nakano2026optimization}.
The agreement-flow equality is imposed separately.
\end{remark}

\section{Optimality Preservation of the Safe Gradient Flow}
\label{sec:optimality-preservation}

This section studies preservation of the optimal solution on the ideal regulated manifold, where $\chi_i=\chi_i^\star(g_i)$ and $y_i=g_i$. In this case, the CBF constraints can be compared directly with the KKT conditions of \eqref{eq:static-distributed-optimization}. The tracking-error perturbation is
handled later in the convergence analysis.

\subsection{Ideal regulated flow and steady-state KKT conditions}
\label{subsec:ideal-regulated-flow-kkt}

Write the agreement constraint as $L_\otimes^\top r=0$, with
$r:=\col(r_1,\ldots,r_N)$. A point
$r^\star=\col(r_1^\star,\ldots,r_N^\star)$ satisfies the KKT conditions of
\eqref{eq:static-distributed-optimization} if there exist multipliers
$\mu^\star$ and $\lambda_{ik}^\star\ge0$ such that, for all $i,k$,
\begin{equation}
\begin{aligned}
    &
    \nabla c_i(r_i^\star)
    +
    (L_\otimes\mu^\star)_i
    -
    \sum_{k\in\mathcal K_i}
    \lambda_{ik}^\star\nabla\eta_{ik}(r_i^\star)
    =0,\\
    &
    (L_\otimes^\top r^\star)_i=0,
    \qquad
    \eta_{ik}(r_i^\star)\ge0,\\
    &
    \lambda_{ik}^\star\ge0,
    \qquad
    \lambda_{ik}^\star\eta_{ik}(r_i^\star)=0 .
\end{aligned}
\label{eq:original-kkt}
\end{equation}
Here $(\cdot)_i$ denotes the $i$th $p$-dimensional block. Under
Assumption~\ref{ass:graph}, $L_\otimes^\top r^\star=0$ implies agreement.

We now evaluate the constraints of the safe gradient flow on the regulated
manifold. Define the steady-state values and reference normals of the safety
DSMs by
\begin{equation}
\begin{aligned}
    \bar{\mathsf M}_{ik}^{\mathsf S}(g_i)
    &:={\mathsf M}_{ik}^{\mathsf S}\big(\chi_i^\star(g_i),g_i\big),\\
    \bar n_{ik}^{\mathsf S}(g_i)
    &:=
    \nabla_{g_i}{\mathsf M}_{ik}^{\mathsf S}
    \big(\chi_i^\star(g_i),g_i\big).
\end{aligned}
\label{eq:ss-safety-dsm-value-normal}
\end{equation}
where $\nabla_{g_i}$ denotes the partial gradient with respect to the
reference argument. Similarly, for agents with a stability-domain DSM, define
\begin{equation}
\begin{aligned}
    \bar{\mathsf M}_{i}^{\mathsf D}(g_i)
    &:={\mathsf M}_{i}^{\mathsf D}\big(\chi_i^\star(g_i),g_i\big),\\
    \bar n_{i}^{\mathsf D}(g_i)
    &:=
    \nabla_{g_i}{\mathsf M}_{i}^{\mathsf D}
    \big(\chi_i^\star(g_i),g_i\big).
\end{aligned}
\label{eq:ss-domain-dsm-value-normal}
\end{equation}

The ideal regulated safe gradient flow is obtained from
\eqref{eq:dsgf-qp-local} by setting $y_i=g_i$,
$\chi_i=\chi_i^\star(g_i)$, and hence $\dot\chi_i=0$. It is the following QP:
\begingroup
\scriptsize
\setlength{\jot}{1pt}
\begin{equation}
\begin{aligned}
    \rho^{\rm id}
    =&
    \arg\min_{\rho}\quad
        \frac12\sum_{i=1}^N
        \left\|\rho_i+\nabla c_i(g_i)\right\|^2\\
    \mathrm{s.t.}\quad&
        \sum_{j\in\mathcal N}a_{ji}(\rho_i-\rho_j)
        +
        \alpha_c\sum_{j\in\mathcal N}a_{ji}(g_i-g_j)=0,
        \quad i\in\mathcal N,\\
    &
        \nabla\eta_{ik}(g_i)^\top\rho_i
        +
        \alpha_{\eta}\eta_{ik}(g_i)
        \ge0,
        \quad i\in\mathcal N,\ k\in\mathcal K_i,\\
    &
        \bar n_{ik}^{\mathsf S}(g_i)^\top\rho_i
        +
        \alpha_{\mathsf S}
        \bar{\mathsf M}_{ik}^{\mathsf S}(g_i)
        \ge0,
        \quad i\in\mathcal N,\ k\in\mathcal K_i,\\
    &
        \bar n_i^{\mathsf D}(g_i)^\top\rho_i
        +
        \alpha_{\mathsf D}
        \bar{\mathsf M}_i^{\mathsf D}(g_i)
        \ge0,
        \quad i\in\mathcal N_{\mathsf D}.
\end{aligned}
\label{eq:ideal-regulated-sgf}
\end{equation}
\endgroup
The last constraint is omitted for agents with global tracking Lyapunov
functions; \eqref{eq:ideal-regulated-sgf} is only the regulated-manifold
version of \eqref{eq:dsgf-qp-local}.

Define
$\mathcal A_i^\eta(g_i):=\{k\in\mathcal K_i:\eta_{ik}(g_i)=0\}$ and
$\mathcal A_i^{\mathsf S}(g_i):=\{k\in\mathcal K_i:
\bar{\mathsf M}_{ik}^{\mathsf S}(g_i)=0\}$. A regulated reference $g^\star$ is
DSM-compatible if $\bar{\mathsf M}_{ik}^{\mathsf S}(g_i^\star)\ge0$ and, for
each $k\in\mathcal A_i^{\mathsf S}(g_i^\star)$,
\begin{equation}
    \bar n_{ik}^{\mathsf S}(g_i^\star)
    \in
    \cone
    \left\{
        \nabla\eta_{i\ell}(g_i^\star):
        \ell\in\mathcal A_i^\eta(g_i^\star)
    \right\}.
    \label{eq:safety-dsm-normal-compatibility}
\end{equation}
For the stability-domain DSM, impose the standing regulated-state condition
\begin{equation}
    \bar{\mathsf M}_i^{\mathsf D}(g_i^\star)>0,
    \qquad i\in\mathcal N_{\mathsf D}.
    \label{eq:domain-dsm-strict-inactivity-condition}
\end{equation}
This condition says that the Lyapunov-domain margin is not an additional
steady-state optimization constraint; if the tracking Lyapunov function is
global, the stability-domain DSM is omitted.

\begin{theorem}[Ideal equilibrium--KKT equivalence]
\label{thm:ideal-equilibrium-kkt}
Suppose the QP in \eqref{eq:ideal-regulated-sgf} is regular at the considered
point, and let $g^\star$ be DSM-compatible in the sense of
\eqref{eq:safety-dsm-normal-compatibility}--\eqref{eq:domain-dsm-strict-inactivity-condition}.
Then $\rho^{\rm id}(g^\star)=0$ if and only if
$r_i^\star=g_i^\star$ satisfies the KKT conditions
\eqref{eq:original-kkt} of the static problem
\eqref{eq:static-distributed-optimization}. Consequently, under
Assumption~\ref{ass:cost}, if the feasible set is convex, the equilibrium
reference agrees with the unique optimal solution $r^\star$.
\end{theorem}

\begin{proof}
Suppose first that $\rho^{\rm id}(g^\star)=0$. Then
$L_\otimes^\top g^\star=0$ and $\eta_{ik}(g_i^\star)\ge0$. Let
$\mu$, $\lambda_{ik}\ge0$, $\nu_{ik}^{\mathsf S}\ge0$, and
$\nu_i^{\mathsf D}\ge0$ be the QP multipliers. Stationarity at $\rho=0$ gives
\begin{equation}
\begin{aligned}
    0
    =
    &
    \nabla c_i(g_i^\star)
    +
    (L_\otimes\mu)_i
    -
    \sum_{k\in\mathcal K_i}
    \lambda_{ik}\nabla\eta_{ik}(g_i^\star)\\
    &-
    \sum_{k\in\mathcal K_i}
    \nu_{ik}^{\mathsf S}
    \bar n_{ik}^{\mathsf S}(g_i^\star)
    -
    \nu_i^{\mathsf D}
    \bar n_i^{\mathsf D}(g_i^\star).
\end{aligned}
\label{eq:ideal-regulated-stationarity}
\end{equation}
By complementarity, $\nu_i^{\mathsf D}=0$ and inactive safety-DSM multipliers
vanish. For active safety DSMs, compatibility gives $\beta_{ik\ell}\ge0$ such
that
\[
    \bar n_{ik}^{\mathsf S}(g_i^\star)
    =
    \sum_{\ell\in\mathcal A_i^\eta(g_i^\star)}
    \beta_{ik\ell}
    \nabla\eta_{i\ell}(g_i^\star).
\]
Absorb the active safety-DSM terms into the steady-state safety multipliers by
defining
\[
    \tilde\lambda_{i\ell}
    :=
    \lambda_{i\ell}
    +
    \sum_{k\in\mathcal A_i^{\mathsf S}(g_i^\star)}
    \nu_{ik}^{\mathsf S}\beta_{ik\ell},
    \qquad
    \ell\in\mathcal A_i^\eta(g_i^\star),
\]
and $\tilde\lambda_{i\ell}:=\lambda_{i\ell}$ otherwise. These multipliers are
nonnegative and preserve complementarity, so
\eqref{eq:ideal-regulated-stationarity} reduces to the stationarity condition
in \eqref{eq:original-kkt}.

Conversely, if $r^\star$ satisfies \eqref{eq:original-kkt}, set
$g_i^\star=r_i^\star$ and $\rho_i=0$. The constraints of
\eqref{eq:ideal-regulated-sgf} are feasible by the KKT feasibility,
DSM-compatibility, and \eqref{eq:domain-dsm-strict-inactivity-condition}. With
zero DSM multipliers, the KKT multipliers of \eqref{eq:original-kkt} satisfy
the QP KKT conditions at $\rho=0$; strong convexity in $\rho$ makes it the
unique optimal solution.
\end{proof}

Theorem~\ref{thm:ideal-equilibrium-kkt} shows that the steady-state CBFs encode
the original admissible set, while compatible safety DSMs and strictly inactive
stability-domain DSMs do not alter the optimal solution.

\subsection{Reference-space compatibility of safety DSMs}
\label{subsec:reference-space-compatibility-safety-dsm}

We now verify \eqref{eq:safety-dsm-normal-compatibility} for the exact safety
DSM $\mathsf M_{ik}^{\mathsf S}=\Gamma_{ik}^{\mathsf S}-V_i$.

\begin{proposition}[Reference-space compatibility of the exact safety DSM]
\label{prop:exact-safety-dsm-compatibility}
Fix an agent $i$ and a constraint $k\in\mathcal K_i$.

(i) If $\eta_{ik}(g_i)>0$, then
\[
    \bar{\mathsf M}_{ik}^{\mathsf S}(g_i)>0,
\]
provided the lifted safety boundary is nonempty. If the lifted boundary is
empty, the corresponding safety DSM constraint is inactive.

(ii) Let $g_i^\star$ satisfy $\eta_{ik}(g_i^\star)=0$. Suppose
$\Gamma_{ik}^{\mathsf S}$ is differentiable at $g_i^\star$,
$\chi_i^\star(g_i)$ is differentiable at $g_i^\star$, and
$\chi_i^\star(g_i^\star)$ is an interior minimizer of
$V_i(\cdot,g_i^\star)$. Suppose also that the first-order necessary condition
holds for the reference-space problem
\[
    \min_{g_i}\ \Gamma_{ik}^{\mathsf S}(g_i)
    \quad
    \mathrm{s.t.}\quad
    \eta_{i\ell}(g_i)\ge0,\ \ell\in\mathcal K_i,
\]
at $g_i^\star$; for example, this holds for the active reference
constraints. Then
\begin{equation}
    \bar{\mathsf M}_{ik}^{\mathsf S}(g_i^\star)=0,
    \label{eq:ss-safety-dsm-zero-value}
\end{equation}
and
\begin{equation}
    \bar n_{ik}^{\mathsf S}(g_i^\star)
    \in
    \cone
    \left\{
        \nabla\eta_{i\ell}(g_i^\star):
        \ell\in\mathcal A_i^\eta(g_i^\star)
    \right\}.
    \label{eq:ss-safety-dsm-normal-cone}
\end{equation}
\end{proposition}

\begin{proof}
First suppose $\eta_{ik}(g_i)>0$. Since
$h_i(\chi_i^\star(g_i))=g_i$, continuity gives $\delta_{ik}(g_i)>0$ such that
all lifted-boundary points lie at least $\delta_{ik}(g_i)$ away from
$\chi_i^\star(g_i)$. Hence, by \eqref{eq:agent-tracking-lyapunov-bounds},
\[
    \Gamma_{ik}^{\mathsf S}(g_i)
    \ge
    \alpha_{i,1}\big(\delta_{ik}(g_i)\big)
    >
    0 .
\]
Since $V_i(\chi_i^\star(g_i),g_i)=0$, this gives
$\bar{\mathsf M}_{ik}^{\mathsf S}(g_i)>0$.

Now suppose $\eta_{ik}(g_i^\star)=0$. Then
$\chi_i^\star(g_i^\star)$ lies on the lifted safety boundary. Since it is
feasible for the threshold problem and
$V_i(\chi_i^\star(g_i^\star),g_i^\star)=0\le V_i(\chi_i,g_i^\star)$, we obtain
$\Gamma_{ik}^{\mathsf S}(g_i^\star)=0$ and hence
$\bar{\mathsf M}_{ik}^{\mathsf S}(g_i^\star)=0$.

Moreover, $\Gamma_{ik}^{\mathsf S}\ge0$ on $\mathcal G_i$ and vanishes at
$g_i^\star$, so $g_i^\star$ is a local minimizer of the reference-space
problem in the proposition. The assumed first-order necessary condition gives
\begin{equation}
    \nabla_{g_i}\Gamma_{ik}^{\mathsf S}(g_i^\star)
    \in
    \cone
    \left\{
        \nabla\eta_{i\ell}(g_i^\star):
        \ell\in\mathcal A_i^\eta(g_i^\star)
    \right\}.
    \label{eq:gamma-safety-cone}
\end{equation}

Finally, $V_i(\chi_i^\star(g_i),g_i)=0$ near $g_i^\star$ and
$\chi_i^\star(g_i^\star)$ is an interior minimizer, hence
$\nabla_{g_i}V_i(\chi_i^\star(g_i^\star),g_i^\star)=0$. Therefore,
\[
\begin{aligned}
    \bar n_{ik}^{\mathsf S}(g_i^\star)
    &=
    \nabla_{g_i}\mathsf M_{ik}^{\mathsf S}
    \big(\chi_i^\star(g_i^\star),g_i^\star\big)\\
    &=
    \nabla_{g_i}\Gamma_{ik}^{\mathsf S}(g_i^\star)
    -
    \nabla_{g_i}V_i
    \big(\chi_i^\star(g_i^\star),g_i^\star\big)\\
    &=
    \nabla_{g_i}\Gamma_{ik}^{\mathsf S}(g_i^\star).
\end{aligned}
\]
Combining this identity with \eqref{eq:gamma-safety-cone} proves
\eqref{eq:ss-safety-dsm-normal-cone}.
\end{proof}

The proposition verifies DSM compatibility for exact safety DSMs; the
stability-domain DSM is handled by the strict-inactivity condition
\eqref{eq:domain-dsm-strict-inactivity-condition}.

\section{Convergence Analysis via a Lyapunov Small-Gain Argument}
\label{sec:convergence}

This section proves convergence through a small-gain argument between the
tracking layer and the upper safe gradient flow \cite{jiang1994small}. Let
$\chi^\star(g):=\col(\chi_1^\star(g_1),\ldots,\chi_N^\star(g_N))$,
$e_\chi:=\chi-\chi^\star(g)$, and $C(g):=\sum_i c_i(g_i)$. The key step is to
compare the implemented QP with the regulated-manifold QP
\eqref{eq:ideal-regulated-sgf}, showing that tracking error perturbs the QP
solution proportionally to $\|e_\chi\|$. For notation, take the common rate
$\alpha_c=\alpha_\eta=\alpha_{\mathsf S}=\alpha_{\mathsf D}=:\alpha_o>0$;
different rates can be handled by rescaling constraints.

\subsection{Ideal upper-flow value function}
\label{subsec:ideal-upper-value-function}

Let $E$ be full row rank with $Eg=0\Longleftrightarrow L_\otimes^\top g=0$,
and let $A^{\rm id}(g)\rho\ge b^{\rm id}(g)$ collect the regulated inequalities
of \eqref{eq:ideal-regulated-sgf}. Then the regulated flow is
\begin{equation}
\begin{aligned}
    \rho^{\rm id}(g)
    =
    \arg\min_{\rho}\quad&
        \frac12\|\rho+\nabla C(g)\|^2\\
    \mathrm{s.t.}\quad&
        E\rho+\alpha_o Eg=0,\\
    &
        A^{\rm id}(g)\rho\ge b^{\rm id}(g).
\end{aligned}
\label{eq:ideal-upper-compact}
\end{equation}
Following the value-function construction for safe gradient flows, define
\begin{equation}
    \mathcal W_o(g)
    :=
    \inf_{\rho}
    \left\{
        \alpha_o C(g)
        +
        \nabla C(g)^\top\rho
        +
        \frac12\|\rho\|^2
    \right\}
    \label{eq:upper-qp-value-function}
\end{equation}
subject to the same constraints as in \eqref{eq:ideal-upper-compact}. Since
$\alpha_o C(g)$ is independent of $\rho$, the minimizer is still
$\rho^{\rm id}(g)$ and
\begin{equation}
    \mathcal W_o(g)
    =
    \alpha_o C(g)
    +
    \nabla C(g)^\top\rho^{\rm id}(g)
    +
    \frac12\|\rho^{\rm id}(g)\|^2 .
    \label{eq:upper-qp-value-evaluated}
\end{equation}
On the compact set below, shift $\mathcal W_o$ by a constant and denote the
shifted function by $V_o$.

\begin{assumption}[Upper-flow value function]
\label{ass:upper-value}
On the compact reference set under consideration, the value function $V_o$ is
continuously differentiable and there exist constants $a_o,b_o>0$ such that
\begin{equation}
    \nabla V_o(g)^\top\rho^{\rm id}(g)
    \le
    -a_o\|\rho^{\rm id}(g)\|^2,
    \label{eq:upper-value-dissipation}
\end{equation}
and
\begin{equation}
    \|\nabla V_o(g)\|
    \le
    b_o\|\rho^{\rm id}(g)\|.
    \label{eq:upper-value-gradient-bound}
\end{equation}
\end{assumption}

A standard sufficient condition is regularity of the KKT solution of
\eqref{eq:ideal-upper-compact} and $\alpha_o I-Q_o(g)\succeq a_oI$ on the
compact set, where $Q_o(g)$ is the active-constraint Lagrangian Hessian. Then
$\nabla V_o(g)=-(\alpha_o I-Q_o(g))\rho^{\rm id}(g)$, giving
\eqref{eq:upper-value-dissipation}--\eqref{eq:upper-value-gradient-bound}.

\subsection{QP perturbation induced by tracking error}
\label{subsec:qp-perturbation-tracking-error}

We next compare the implemented QP \eqref{eq:dsgf-qp-local} with its
regulated-manifold counterpart \eqref{eq:ideal-upper-compact}. The point is to
show that the measured output and the transient DSM evaluations perturb the
regulated QP only through the tracking error $e_\chi$.

Let $\nabla C_y(\chi):=\col(\nabla c_1(h_1(\chi_1)),\ldots,\nabla c_N(h_N(\chi_N)))$. For this perturbation estimate only, collect the inequality rows of the implemented QP in the affine form $A(\chi,g)\rho\ge b(\chi,g)$, using the same row order as $A^{\rm id}(g)\rho\ge b^{\rm id}(g)$. With this notation, the implemented QP is equivalently written as
\begin{equation}
\begin{aligned}
    \rho^\star(\chi,g)
    =
    \arg\min_{\rho}\quad&
        \frac12\|\rho\|^2+\nabla C_y(\chi)^\top\rho\\
    \mathrm{s.t.}\quad&
        E\rho+\alpha_oEg=0,\\
    &
        A(\chi,g)\rho\ge b(\chi,g).
\end{aligned}
\label{eq:actual-qp-standard-form}
\end{equation}
This is \eqref{eq:dsgf-qp-local} after dropping the constant term
$\frac12\|\nabla C_y(\chi)\|^2$ in the objective. By construction,
$A^{\rm id}(g)\rho\ge b^{\rm id}(g)$ is obtained from
$A(\chi,g)\rho\ge b(\chi,g)$ by setting $\chi=\chi^\star(g)$ and
$\dot\chi=0$.

\begin{lemma}[Perturbation from the regulated QP]
\label{lem:rho-perturbation}
Suppose that, on the compact set under consideration, the maps
$h_i$, $\nabla c_i$, $\mathsf M_{ik}^{\mathsf S}$,
$\mathsf M_i^{\mathsf D}$, and the gradients entering the QP constraints are
locally Lipschitz. Suppose also that $\dot\chi_i=0$ at
$\chi_i=\chi_i^\star(g_i)$ and that $\dot\chi_i$ is locally Lipschitz in
$\chi_i$ uniformly in $g_i$. Finally, assume that the KKT generalized equation
of the regulated QP \eqref{eq:ideal-upper-compact} is uniformly strongly
regular on the compact set. Then there exists $\ell_\rho>0$ such that
\begin{equation}
    \|\rho^\star(\chi,g)-\rho^{\rm id}(g)\|
    \le
    \ell_\rho\|\chi-\chi^\star(g)\|.
    \label{eq:rho-star-rho-id-bound}
\end{equation}
\end{lemma}

\begin{proof}
The equality constraint is identical in \eqref{eq:actual-qp-standard-form} and
\eqref{eq:ideal-upper-compact}. By Lipschitz continuity of $h_i$ and
$\nabla c_i$, and using $h(\chi^\star(g))=g$, there exists $L_c>0$ such that
\begin{equation}
    \|\nabla C_y(\chi)-\nabla C(g)\|
    \le L_c\|\chi-\chi^\star(g)\| .
\label{eq:objective-vector-perturbation}
\end{equation}
The steady-state CBF rows depend only on $g_i$ and are unchanged. Local
Lipschitzness of the DSM gradients gives, after stacking all DSM rows,
\begin{equation}
    \|A(\chi,g)-A^{\rm id}(g)\|
    \le L_A\|\chi-\chi^\star(g)\| .
    \label{eq:A-data-bound}
\end{equation}
For the right-hand side, the implemented DSM-CBF rows contain
$-\nabla_{\chi_i}\mathsf M(\chi_i,g_i)^\top\dot\chi_i-
\alpha_o\mathsf M(\chi_i,g_i)$, while the regulated rows contain
$-\alpha_o\bar{\mathsf M}(g_i)$. Since $\dot\chi_i=0$ at
$\chi_i=\chi_i^\star(g_i)$ and the involved maps are locally Lipschitz and
bounded on the compact set, stacking all rows gives
\begin{equation}
    \|b(\chi,g)-b^{\rm id}(g)\|
    \le L_b\|\chi-\chi^\star(g)\| .
    \label{eq:b-data-bound}
\end{equation}
Consequently, for some $L_\Delta>0$,
\begin{equation}
\begin{aligned}
    &\|\nabla C_y(\chi)-\nabla C(g)\|
    +\|A(\chi,g)-A^{\rm id}(g)\|
    +\|b(\chi,g)-b^{\rm id}(g)\| \\
    &\hspace{7em}\le L_\Delta\|\chi-\chi^\star(g)\| .
\end{aligned}
\label{eq:qp-coefficient-perturbation-bound}
\end{equation}
Uniform strong regularity of the regulated QP KKT generalized equation implies
Lipschitz dependence of the primal solution on these data, uniformly over the
compact set. Hence
$\|\rho^\star(\chi,g)-\rho^{\rm id}(g)\|\le
L_{\rm K}L_\Delta\|\chi-\chi^\star(g)\|$, so
\eqref{eq:rho-star-rho-id-bound} holds with
$\ell_\rho:=L_{\rm K}L_\Delta$.
\end{proof}

Thus the implemented safe-gradient velocity differs from the regulated one
only by a tracking-error perturbation.

\subsection{Tracking estimation under moving references}
\label{subsec:tracking-iss-estimate}

We now derive an input-to-state estimation for the tracking layer. For each
fixed $g_i\in\mathcal G_i$, the Lyapunov function $V_i(\chi_i,g_i)$ satisfies
\begin{equation}
    \nabla_{\chi_i}V_i(\chi_i,g_i)^\top\mathcal F_i^{\rm cl}(\chi_i,g_i)
    \le
    -W_i(\chi_i,g_i).
    \label{eq:frozen-reference-lyapunov-recall}
\end{equation}
When $\dot g_i=\rho_i^\star$, the additional term
$\nabla_{g_i}V_i(\chi_i,g_i)^\top\rho_i^\star$ enters as an input.

On the compact set under consideration, assume there exist
$\underline w_i,\bar v_i>0$ such that
\begin{equation}
    W_i(\chi_i,g_i)
    \ge
    \underline w_i
    \|\chi_i-\chi_i^\star(g_i)\|^2,
    \label{eq:tracking-W-quadratic-bound}
\end{equation}
and
\begin{equation}
    \|\nabla_{g_i}V_i(\chi_i,g_i)\|
    \le
    \bar v_i
    \|\chi_i-\chi_i^\star(g_i)\|.
    \label{eq:tracking-g-gradient-bound}
\end{equation}
The second bound follows locally from smoothness when the Lyapunov functions
are centered at interior minimizers $\chi_i^\star(g_i)$.

Let
$
    V_x(\chi,g):=\sum_{i=1}^N V_i(\chi_i,g_i),
    e_\chi:=\chi-\chi^\star(g).
$
Combining the preceding bounds and applying Young's inequality, for any
$\varepsilon_x\in(0,1)$,
\begin{equation}
    \dot V_x
    \le
    -a_x\|e_\chi\|^2
    +
    b_x\|\rho^\star\|^2,
    \label{eq:tracking-iss}
\end{equation}
where one may take
$
    a_x:=(1-\varepsilon_x)\min_{i\in\mathcal N}\underline w_i,
    b_x:=\max_{i\in\mathcal N}
    \frac{\bar v_i^2}{4\varepsilon_x\underline w_i}.
$
This estimate can be rewritten in terms of the ideal regulated velocity
$\rho^{\rm id}(g)$.

\begin{lemma}[Tracking estimation relative to the ideal flow]
\label{lem:tracking-estimate-ideal}
Suppose the bounds \eqref{eq:tracking-W-quadratic-bound} and
\eqref{eq:tracking-g-gradient-bound} hold on the compact set under
consideration, and suppose Lemma~\ref{lem:rho-perturbation} holds. Define
$\bar a_x:=a_x-2b_x\ell_\rho^2$ and $\bar b_x:=2b_x$. If
$\bar a_x>0$, then
\begin{equation}
    \dot V_x
    \le
    -\bar a_x\|e_\chi\|^2
    +
    \bar b_x\|\rho^{\rm id}(g)\|^2 .
    \label{eq:tracking-ideal-estimate}
\end{equation}
\end{lemma}

\begin{proof}
By Lemma~\ref{lem:rho-perturbation},
$
    \rho^\star
    =
    \rho^{\rm id}
    +
    \big(\rho^\star-\rho^{\rm id}\big),
    \|\rho^\star-\rho^{\rm id}\|
    \le
    \ell_\rho\|e_\chi\|.
$
Thus,
\begin{equation}
    \|\rho^\star\|^2
    \le
    2\|\rho^{\rm id}\|^2
    +
    2\ell_\rho^2\|e_\chi\|^2 .
    \label{eq:rho-star-bound-by-rho-id}
\end{equation}
Substituting \eqref{eq:rho-star-bound-by-rho-id} into
\eqref{eq:tracking-iss} gives
\[
    \dot V_x
    \le
    -
    \big(a_x-2b_x\ell_\rho^2\big)\|e_\chi\|^2
    +
    2b_x\|\rho^{\rm id}(g)\|^2,
\]
which is exactly \eqref{eq:tracking-ideal-estimate}.
\end{proof}

\subsection{Small-gain convergence and optimality recovery}
\label{subsec:small-gain-convergence}

We now combine the upper-flow estimate and the tracking estimate. Along the
implemented closed-loop system, $\dot g=\rho^\star(\chi,g)$. Hence,
\begin{equation}
\begin{aligned}
    \dot V_o
    &=
    \nabla V_o(g)^\top\rho^\star\\
    &=
    \nabla V_o(g)^\top\rho^{\rm id}
    +
    \nabla V_o(g)^\top
    \big(\rho^\star-\rho^{\rm id}\big).
\end{aligned}
\label{eq:upper-value-actual-derivative}
\end{equation}
Using Assumption~\ref{ass:upper-value} and
Lemma~\ref{lem:rho-perturbation}, we obtain
\begin{equation}
    \dot V_o
    \le
    -a_o\|\rho^{\rm id}\|^2
    +
    b_o\ell_\rho
    \|\rho^{\rm id}\|\|e_\chi\|.
    \label{eq:upper-flow-before-young}
\end{equation}
For any $\varepsilon_o\in(0,1)$, Young's inequality gives
\begin{equation}
    \dot V_o
    \le
    -(1-\varepsilon_o)a_o\|\rho^{\rm id}\|^2
    +
    d_o\|e_\chi\|^2,
    \label{eq:upper-flow-iss}
\end{equation}
where
$
    d_o:=
    \frac{b_o^2\ell_\rho^2}{4\varepsilon_o a_o}.
$

Combining \eqref{eq:upper-flow-iss} with \eqref{eq:tracking-ideal-estimate}
for $\mathcal V(\chi,g):=V_x(\chi,g)+\gamma V_o(g)$, $\gamma>0$, gives
\begin{equation}
\begin{aligned}
    \dot{\mathcal V}
    \le&
    -
    \big(\bar a_x-\gamma d_o\big)\|e_\chi\|^2\\
    &-
    \big(\gamma(1-\varepsilon_o)a_o-\bar b_x\big)
    \|\rho^{\rm id}(g)\|^2 .
\end{aligned}
\label{eq:composite-lyap-derivative}
\end{equation}
Thus $\mathcal V$ decreases if there exists $\gamma>0$ such that
\begin{equation}
    \frac{\bar b_x}{(1-\varepsilon_o)a_o}
    <
    \gamma
    <
    \frac{\bar a_x}{d_o}.
    \label{eq:small-gain-gamma-interval}
\end{equation}
Equivalently,
\begin{equation}
    \bar b_xd_o
    <
    (1-\varepsilon_o)a_o\bar a_x .
    \label{eq:small-gain-condition}
\end{equation}
This is the small-gain condition between the two layers.

\begin{theorem}[Small-gain convergence and optimality recovery]
\label{thm:small-gain-convergence}
Suppose Assumptions~\ref{ass:cost} and \ref{ass:graph} hold. Suppose also that the hypotheses
of Lemma~\ref{lem:rho-perturbation} and
Lemma~\ref{lem:tracking-estimate-ideal} hold on a compact forward-invariant set
containing the closed-loop trajectory. Let $\varepsilon_o\in(0,1)$ be fixed,
and assume that $\bar a_x=a_x-2b_x\ell_\rho^2>0$ and that the small-gain
condition \eqref{eq:small-gain-condition} holds. If the initial condition
satisfies
$g_i(0)\in\mathcal G_i$,
$\mathsf M_{ik}^{\mathsf S}(\chi_i(0),g_i(0))\ge0$, and
$\mathsf M_i^{\mathsf D}(\chi_i(0),g_i(0))\ge0$ for all relevant $i$ and $k$,
then the closed-loop trajectory satisfies
\begin{equation}
    \eta_{ik}(y_i(t))\ge0,
    \quad
    \forall t\ge0,\ i\in\mathcal N,\ k\in\mathcal K_i,
    \label{eq:closed-loop-output-safety}
\end{equation}
and
\begin{equation}
    e_\chi(t)\to0,
    \quad
    \rho^{\rm id}(g(t))\to0,
    \quad
    \rho^\star(\chi(t),g(t))\to0 .
    \label{eq:convergence-e-rho}
\end{equation}
Moreover,
\begin{equation}
    y_i(t)-g_i(t)\to0,
    \qquad
    i\in\mathcal N,
    \label{eq:output-reference-tracking-convergence}
\end{equation}
and
\begin{equation}
    L_\otimes^\top g(t)\to0 .
    \label{eq:consensus-convergence}
\end{equation}
Every limit point $(\bar\chi,\bar g)$ of the closed-loop trajectory satisfies
\begin{equation}
    \bar\chi_i=\chi_i^\star(\bar g_i),
    \quad
    \rho^{\rm id}(\bar g)=0,
    \quad
    L_\otimes^\top\bar g=0 .
    \label{eq:limit-point-ideal-equilibrium}
\end{equation}
If, in addition, $\bar g$ satisfies the DSM-compatibility condition
\eqref{eq:safety-dsm-normal-compatibility} and the regulated
stability-domain strict-inactivity condition
\eqref{eq:domain-dsm-strict-inactivity-condition}, then $\bar g$ satisfies the
KKT conditions \eqref{eq:original-kkt} of
\eqref{eq:static-distributed-optimization}. If all limit points satisfy these
conditions and the feasible set $\mathcal S$ is convex, then the static
problem has the unique optimal solution $r^\star$, and
\begin{equation}
    \lim_{t\to\infty}g_i(t)=r^\star,
    \quad
    \lim_{t\to\infty}y_i(t)=r^\star,
    \quad
    i\in\mathcal N .
    \label{eq:final-output-convergence}
\end{equation}
\end{theorem}

\begin{proof}
The safety claim follows from the CBF comparison argument in
Section~\ref{sec:dsgf}: the reference set and DSM sets are forward invariant,
so the transient output-safety certificates remain valid.

By the small-gain condition \eqref{eq:small-gain-condition}, one can choose
$\gamma$ satisfying \eqref{eq:small-gain-gamma-interval}. Then the composite
Lyapunov estimate \eqref{eq:composite-lyap-derivative} yields constants
$c_x,c_o>0$ such that
$\dot{\mathcal V}\le -c_x\|e_\chi\|^2-c_o\|\rho^{\rm id}(g)\|^2$.
Since the trajectory remains in a compact set and $\mathcal V$ is bounded from
below, it follows that $e_\chi\in L_2$ and $\rho^{\rm id}(g)\in L_2$. The
closed-loop vector field is locally Lipschitz on the compact set, so
$e_\chi(t)$ and $\rho^{\rm id}(g(t))$ are uniformly continuous. Barbalat's
lemma gives $e_\chi(t)\to0$ and $\rho^{\rm id}(g(t))\to0$. Lemma~\ref{lem:rho-perturbation}
then gives $\rho^\star(\chi(t),g(t))\to0$, proving
\eqref{eq:convergence-e-rho}.

Since $h_i$ is locally Lipschitz and $h_i(\chi_i^\star(g_i))=g_i$,
$e_\chi(t)\to0$ implies $y_i(t)-g_i(t)\to0$. Also,
$L_\otimes^\top\rho^\star+\alpha_cL_\otimes^\top g=0$ and
$\dot g=\rho^\star$ give $L_\otimes^\top g(t)\to0$.

For any limit point $(\bar\chi,\bar g)$, the preceding limits imply
\eqref{eq:limit-point-ideal-equilibrium}. Theorem~\ref{thm:ideal-equilibrium-kkt}
then gives the KKT conditions whenever the DSM-compatibility and strict-inactivity
conditions hold. If all limit points satisfy them and $\mathcal S$ is convex,
Assumption~\ref{ass:cost} gives the unique optimal solution $r^\star$, so
$g_i(t)\to r^\star$ and, because $y_i(t)-g_i(t)\to0$, also
$y_i(t)\to r^\star$.
\end{proof}

\section{Adaptive tangential objective shaping for nonconvex safety constraints}

We discuss a preliminary extension to nonconvex reference-space safety
constraints. The implemented reference governor still contains the
equality-flow constraint. Hence, at any equilibrium of the reference dynamics,
\(\rho^\star=0\) implies \(L_\otimes^\top g=0\), and therefore
\(g_i=r\) for all \(i\in\mathcal N\). Thus, the local equilibrium analysis
below is carried out on the agreement coordinate \(r\), where the objective is
\(C(r):=\sum_{i=1}^N c_i(r)\).

Fix an agent \(i\in\mathcal N\), and let \(\eta_{i0}(r)\ge0\) denote a
nonconvex reference-space safety constraint associated with agent \(i\). Assume
that \(\eta_{i0}\) is \(C^2\), \(\nabla\eta_{i0}(r)\neq0\) on
\(\eta_{i0}(r)=0\), and the desired optimal solution is strictly feasible, i.e.,
\(\eta_{i0}(r^\star)\ge\delta_{i\star}>0\) for some
\(\delta_{i\star}>0\). We focus on a boundary equilibrium \(r_e\neq r^\star\)
where \(\eta_{i0}(r_e)=0\), all other convex reference-space constraints are
inactive, the stability-domain DSMs are strictly inactive, and active safety
DSMs are compatible with the active reference-safety normal as in
Section~5. Then the reduced KKT condition is
\begin{equation}
\label{eq:nonconvex-reduced-kkt}
    \nabla C(r_e)-\lambda_{i,e}\nabla\eta_{i0}(r_e)=0,
    \quad
    \lambda_{i,e}\ge0,
    \quad
    \eta_{i0}(r_e)=0 .
\end{equation}
This condition describes a nonconvex-induced undesirable equilibrium: the
aggregate descent direction is cancelled by the CBF-induced boundary normal.

Let \(T_{i,e}\in\mathbb R^{p\times(p-1)}\) be an orthonormal basis of
\(\mathcal T_{i,e}:=\{v\in\mathbb R^p:\nabla\eta_{i0}(r_e)^\top v=0\}\).
Define the reduced Hessian
\begin{equation}
\label{eq:nonconvex-reduced-hessian}
    H_{i,e}
    :=
    T_{i,e}^\top
    \left(
        \nabla^2 C(r_e)
        -
        \lambda_{i,e}\nabla^2\eta_{i0}(r_e)
    \right)
    T_{i,e}.
\end{equation}
If \(H_{i,e}\succ0\), the point is a locally attracting constrained minimum of
the reduced safe gradient flow. If \(\lambda_{\min}(H_{i,e})<0\), then there
exists a feasible tangential direction along which the objective can decrease,
and the point is saddle-type \cite{REIS2026113032}.

To destabilize attracting boundary equilibria, introduce an adaptive shaping
state \(\vartheta_i=(\sigma_i,\xi_i,Q_i)\), where \(\sigma_i\ge0\),
\(\xi_i\in\mathbb R^p\), and \(Q_i=Q_i^\top\succeq0\). For
\(\nabla\eta_{i0}(\xi_i)\neq0\), define
\begin{equation}
\label{eq:nonconvex-tangent-projector}
    P_{i,\mathcal T}(\xi_i)
    :=
    I-
    \frac{
        \nabla\eta_{i0}(\xi_i)\nabla\eta_{i0}(\xi_i)^\top
    }{
        \|\nabla\eta_{i0}(\xi_i)\|^2
    } .
\end{equation}
Let \(\chi_i^b(r,\xi_i)\) be a \(C^2\) bump function satisfying
\(0\le\chi_i^b\le1\), \(\chi_i^b(\xi_i,\xi_i)=1\), and
\(\nabla_r\chi_i^b(\xi_i,\xi_i)=0\). Define
\begin{equation}
\label{eq:nonconvex-shaping}
    \Phi_i(r,\vartheta_i)
    :=
    -
    \frac{\sigma_i}{2}
    \chi_i^b(r,\xi_i)
    (r-\xi_i)^\top
    P_{i,\mathcal T}(\xi_i)^\top
    Q_i
    P_{i,\mathcal T}(\xi_i)
    (r-\xi_i).
\end{equation}
In implementation, agent \(i\) replaces its nominal gradient
\(\nabla c_i(y_i)\) in the QP objective by
\(\nabla c_i(y_i)+\nabla_{g_i}\Phi_i(g_i,\vartheta_i)\), while all
agreement-flow, reference-CBF, safety-DSM, and stability-domain DSM
constraints are kept unchanged. On the agreement manifold, this corresponds to
the shaped aggregate objective \(C_i^s(r,\vartheta_i):=C(r)+\Phi_i(r,\vartheta_i)\).

The shaping state evolves continuously as
\begin{equation}
\label{eq:nonconvex-shaping-dynamics}
    \dot\sigma_i
    =
    -k_{\sigma i}(\sigma_i-\bar\sigma_i a_i(g_i)),
    \quad
    \dot\xi_i
    =
    k_{\xi i}a_i(g_i)(g_i-\xi_i),
\end{equation}
where \(k_{\sigma i},k_{\xi i},\bar\sigma_i>0\). The activation
\(a_i(g_i)\in[0,1]\) is smooth, large near \(\eta_{i0}(g_i)=0\) when the
descent direction is nearly aligned with the boundary normal, and zero when
\(\eta_{i0}(g_i)\ge h_{i,\rm off}\), with
\(0<h_{i,\rm off}<\eta_{i0}(r^\star)\). Thus, the shaping term vanishes near the desired optimal solution and is activated through smooth dynamics rather than discontinuous switching.

\begin{proposition}
\label{prop:nonconvex-tangential-shaping}
Let \(r_e\) satisfy \eqref{eq:nonconvex-reduced-kkt} and suppose
\(H_{i,e}\succ0\). If the shaping dynamics admit an equilibrium with
\(\xi_{i,e}=r_e\) and \(\sigma_{i,e}>0\), then \(r_e\) remains a first-order
stationary point of the shaped reduced problem. Moreover, its shaped reduced
Hessian is
\begin{equation}
\label{eq:nonconvex-shaped-hessian}
    H_{i,e}^s
    =
    H_{i,e}
    -
    \sigma_{i,e}T_{i,e}^\top Q_iT_{i,e}.
\end{equation}
Hence, if \(\lambda_{\min}(H_{i,e}^s)<0\), the originally attracting
nonconvex-induced equilibrium becomes saddle-type. In particular, for
\(Q_i=I\), any \(\sigma_{i,e}>\lambda_{\min}(H_{i,e})\) is sufficient.
\end{proposition}

\begin{proof}
Since \(\xi_{i,e}=r_e\), the shaping term is centered at \(r_e\). Together
with \(\nabla_r\chi_i^b(\xi_{i,e},\xi_{i,e})=0\), this gives
\(\nabla_r\Phi_i(r_e,\vartheta_{i,e})=0\). Therefore
\(\nabla C_i^s(r_e,\vartheta_{i,e})=\nabla C(r_e)\), and the first-order KKT
condition \eqref{eq:nonconvex-reduced-kkt} is preserved. On
\(\mathcal T_{i,e}\), one has \(P_{i,\mathcal T}(r_e)T_{i,e}=T_{i,e}\), so the
Hessian contribution of \(\Phi_i\) is
\(-\sigma_{i,e}T_{i,e}^\top Q_iT_{i,e}\). This proves
\eqref{eq:nonconvex-shaped-hessian}. If \(H_{i,e}^s\) has a negative
eigenvalue, the shaped reduced objective has negative curvature along a
feasible tangential direction, so the boundary point is no longer a local
constrained minimum.
\end{proof}

% This construction is local. It does not provide a global optimality theorem
% for general nonconvex feasible sets. If additional convex constraints are
% active at the same boundary point, \(T_{i,e}\) should be replaced by a basis of
% the intersection of all active tangent spaces. If that tangent space is
% zero-dimensional, no smooth tangential escape direction exists locally, and
% hybrid, time-varying, or path-level mechanisms are needed.

\section{Simulation}
\label{sec:simulation}

We consider five heterogeneous aerial vehicles with dynamics adopted from the benchmark example in \cite{hu2015consensus}:
\begin{equation}
\begin{aligned}
    \ddot p_{x,i}-2\omega_i\dot p_{y,i} &= \tau_{x,i},\\
    \ddot p_{y,i}+2\omega_i\dot p_{x,i}-3\omega_i^2p_{y,i} &= \tau_{y,i},\\
    y_i &= \operatorname{col}(p_{x,i},p_{y,i}),
\end{aligned}
\label{eq:sim_original_model}
\end{equation}
where $\omega_i=i-3$. Let $p_i=\operatorname{col}(p_{x,i},p_{y,i})$,
$v_i=\operatorname{col}(\dot p_{x,i},\dot p_{y,i})$, and
$x_i=\operatorname{col}(p_i,v_i)$. Then
\begin{equation}
    \dot p_i=v_i,\qquad
    \dot v_i=D_ip_i+R_iv_i+\tau_i,
\label{eq:sim_state_model}
\end{equation}
where
$
D_i=
\begin{bmatrix}
0&0\\
0&3\omega_i^2
\end{bmatrix},
R_i=
\begin{bmatrix}
0&2\omega_i\\
-2\omega_i&0
\end{bmatrix}.
$
The communication graph is the undirected ring
$\mathcal E=\{(1,2),(2,3),(3,4),(4,5),(5,1)\}$ with unit weights and Laplacian matrix $L$.

The local objective function is $c_i(r)=\frac12\|r-e_i\|^2.$ To test the case where the optimal solution lies on the safety boundary, we set $e_i=e=\operatorname{col}(2.5,2.5)$ for all $i$. The safety set is the regular hexagonal workspace
\begin{equation}
    \Omega=\{r\in\mathbb R^2:\eta_\ell(r)\ge0,\ \ell=1,\ldots,6\},
\end{equation}
where
$
    \eta_\ell(r)=c_\ell+a_\ell^\top r,
    c_\ell=2\cos\frac{\pi}{6}=\sqrt3,
$
and
$
a_\ell
=
-\operatorname{col}
\left(
\cos\left(\frac{\pi}{6}+\frac{(\ell-1)\pi}{3}\right),
\sin\left(\frac{\pi}{6}+\frac{(\ell-1)\pi}{3}\right)
\right).
$
The constrained optimal solution is the Euclidean projection of $e$ onto $\Omega$, namely
$r^\star=\operatorname{col}(1.04,1.66)$. The initial outputs are
$y_1(0)=\operatorname{col}(0,-0.3)$, $y_2(0)=\operatorname{col}(0,0.6)$,
$y_3(0)=\operatorname{col}(-0.3,0.7)$, $y_4(0)=\operatorname{col}(-0.5,-0.1)$,
and $y_5(0)=\operatorname{col}(-1,-0.7)$, with $v_i(0)=0$. In this simulation, only workspace safety is considered; inter-agent collision avoidance is not imposed.

\subsection{Baselines}

\subsubsection{SGF-HOCBF baseline}

We first compare with the SGF-HOCBF feedback optimizer in
\cite{delimpaltadakis2025feedback}. Since the original model lacks a unique
exponentially stable equilibrium for every constant input, we introduce a fixed
pre-stabilizing feedback and use $u_i$ as the SGF-HOCBF optimization input:
\begin{equation}
    \tau_i=-D_ip_i-R_iv_i-K_pp_i-K_dv_i+u_i.
\label{eq:sim_prestabilization_sgf}
\end{equation}
Then the plant seen by SGF-HOCBF is
\begin{equation}
    \dot p_i=v_i,\quad
    \dot v_i=-K_pp_i-K_dv_i+u_i,\quad
    \dot u_i=q_i.
\label{eq:sim_sgf_augmented}
\end{equation}
For constant $u_i$, the steady-state map is
$w_i(u_i)=\operatorname{col}(K_p^{-1}u_i,0)$. The nominal SGF direction is chosen as
\begin{equation}
    q_i^{\rm nom}
    =
    -\epsilon K_p^{-\top}
    \left[p_i-e+\kappa(Lp)_i\right],
\label{eq:sim_sgf_nominal}
\end{equation}
with $\epsilon=0.8$ and $\kappa=0.4$.

For each safety constraint, define $h_{i\ell,0}=\eta_\ell(p_i)$. Since the QP decision variable is $q_i=\dot u_i$, the output safety constraint has high relative degree with respect to $q_i$. We therefore construct the HOCBF chain
$
    h_{i\ell,1}=a_\ell^\top v_i+\beta h_{i\ell,0},
$
and
$
    h_{i\ell,2}
    =
    a_\ell^\top(-K_pp_i-K_dv_i+u_i)
    +2\beta a_\ell^\top v_i
    +\beta^2h_{i\ell,0}.
$
The final HOCBF condition is $\dot h_{i\ell,2}+\gamma h_{i\ell,2}\ge0$, and the SGF-HOCBF QP can be written as
\begin{equation}
\begin{aligned}
q^\star
=&
\arg\min_{q_1,\ldots,q_5}\quad
\frac12\sum_{i=1}^5
\|q_i-q_i^{\rm nom}\|^2 \\
\text{s.t.}\quad
&
\dot h_{i\ell,2}+\gamma h_{i\ell,2}\ge0,
\quad i=1,\ldots,5,\ \ell=1,\ldots,6.
\end{aligned}
\label{eq:sim_sgf_hocbf_qp}
\end{equation}
We use $\beta=3$ and $\gamma=8$. Two pre-stabilizing gains are tested:
$
    K_p=4I, K_d=4I,
$
and
$
    K_p=\operatorname{diag}(1,8),
    K_d=\operatorname{diag}(3,5).
$
The anisotropic choice illustrates possible mismatch between the HOCBF-induced
boundary equilibrium and the original KKT point.

\subsubsection{Projected primal-dual feedback optimization baseline}

We also compare with the projected primal-dual feedback optimization method in
\cite{qin2023distributed}. It handles steady-state inequalities but is not
designed for transient output safety. The hexagonal constraints are written as
$
    Br_i-\bar g\le0,
$
where
$
B=
\operatorname{col}(-a_1^\top, \cdots,-a_6^\top),
\bar g=\operatorname{col}(c_1,\ldots,c_6).
$
We use the same tracking layer as in our method:
$
    \tau_i=-D_ip_i-R_iv_i-K_p(p_i-r_i)-K_dv_i,
$
which gives $\dot p_i=v_i$ and $\dot v_i=-K_p(p_i-r_i)-K_dv_i$. Let
$\chi_i=\operatorname{col}(p_i-r_i,v_i)$. Choose $P=P^\top>0$ satisfying
\[
A_T^\top P+PA_T=-I,
\quad
A_T=
\begin{bmatrix}
0&I_2\\
-K_p&-K_d
\end{bmatrix},
\]
and define $\phi_i=2B_T^\top P\chi_i$ with $B_T=\operatorname{col}(-I_2,0)$. The projected primal-dual feedback optimizer is implemented as
\begin{equation}
\begin{aligned}
\dot r_i
&=
-(p_i-e)-B^\top z_i-\sum_{j\in\mathcal N_i}a_{ij}(\lambda_i-\lambda_j),\\
\dot z_i
&=
[Br_i-\bar g+\sigma B\phi_i]^+_{z_i},\\
\dot \lambda_i
&=
\sum_{j\in\mathcal N_i}a_{ij}(r_i-r_j)
+
\sigma\sum_{j\in\mathcal N_i}a_{ij}(\phi_i-\phi_j),
\end{aligned}
\label{eq:sim_pdfo}
\end{equation}
where $z_i(0)=0$, $\lambda_i(0)=0$, $r_i(0)=p_i(0)$, and $\sigma=0.1$. The projection operator $[\cdot]^+_{z_i}$ keeps $z_i(t)\ge0$, but it does not make $Br_i-\bar g\le0$ or $By_i-\bar g\le0$ forward invariant.

\subsection{Implementation of the Proposed Method}

For the proposed method, we introduce a reference signal $g_i\in\mathbb R^2$ and design $\dot g_i=\rho_i$. The physical tracking controller is chosen from the regulator-equation construction. For the aerial-vehicle model, take
$
\Pi_i=
\begin{bmatrix}
I_2\\
0_{2\times2}
\end{bmatrix},
\qquad
\Psi_i=-D_i.
$
Let
$
K_{1i}=
\begin{bmatrix}
-D_i-K_p & -R_i-K_d
\end{bmatrix},
K_{2i}=\Psi_i-K_{1i}\Pi_i=K_p.
$
Then $\tau_i=K_{1i}x_i+K_{2i}g_i$ is equivalently
\begin{equation}
    \tau_i=-D_ip_i-R_iv_i-K_p(p_i-g_i)-K_dv_i.
\label{eq:sim_tracking_controller}
\end{equation}
The resulting tracking dynamics are
\begin{equation}
    \dot p_i=v_i,\qquad
    \dot v_i=-K_p(p_i-g_i)-K_dv_i.
\label{eq:sim_tracking_dynamics}
\end{equation}
For fixed $g_i$, the equilibrium is $p_i^\star=g_i$, $v_i^\star=0$. Let
$e_i^p=p_i-g_i$ and use the global Lyapunov function
\begin{equation}
    V_i(p_i,v_i,g_i)
    =
    \frac12(e_i^p)^\top K_p e_i^p
    +
    \frac12v_i^\top v_i.
\label{eq:sim_tracking_lyapunov}
\end{equation}
When $g_i$ is fixed, $\dot V_i=-v_i^\top K_dv_i\le0$. Since this Lyapunov function is global, no stability-domain DSM is used.

For each workspace constraint, the exact safety-energy threshold is
\begin{equation}
    \Gamma^S_{i\ell}(g_i)
    =
    \frac12
    \frac{\eta_\ell(g_i)^2}{a_\ell^\top K_p^{-1}a_\ell}.
\label{eq:sim_gamma_safety}
\end{equation}
and the safety DSM is
\begin{equation}
    M^S_{i\ell}(p_i,v_i,g_i)
    =
    \Gamma^S_{i\ell}(g_i)
    -
    V_i(p_i,v_i,g_i).
\label{eq:sim_safety_dsm}
\end{equation}
The reference governor of agent $i$ starts from the local measured
output-gradient direction
$
    \rho_i^{\rm nom}
    =
    -k_g\nabla c_i(y_i)
    =
    -k_g(p_i-e_i).
$
At each time instant, agent $i$ computes its reference velocity by the
following neighbor-coupled distributed QP:
\begin{equation}
\begin{aligned}
\rho_i^\star
=&
\arg\min_{\rho_i}\quad
\frac12
\|\rho_i-\rho_i^{\rm nom}\|^2
\\
\mathrm{s.t.}\quad
&
a_{ij}(\rho_i-\rho_j)
+
\alpha_c a_{ij}(g_i-g_j)
=0,
\qquad j\in\mathcal N_i,
\\
&
a_\ell^\top\rho_i
+
\alpha_\eta\eta_\ell(g_i)
\ge0,
\qquad \ell=1,\ldots,6,
\\
&
\left(
\nabla_{g_i}\Gamma^S_{i\ell}(g_i)
+
K_p(p_i-g_i)
\right)^\top\rho_i
+
v_i^\top K_dv_i
+
\alpha_S M^S_{i\ell}
\ge0,
\\
&\hfill \ell=1,\ldots,6 .
\end{aligned}
\label{eq:sim_dsm_local_qp}
\end{equation}
We use $k_g=1$, $\alpha_c=0.5$, $\alpha_\eta=0.2$, and $\alpha_S=5$, with $g_i(0)=p_i(0)$. The gain choices of isotropic and anisotropic tracking are both tested.

\subsection{Results}

Fig.~\ref{fig:sgf_hocbf_traj} shows the SGF-HOCBF trajectories under the two pre-stabilizing gains. In the isotropic case, the agents approach the true constrained optimal solution $r^\star$, with final average distance about $3.12\times10^{-2}$ and minimum safety value about $-8.88\times10^{-16}$. However, in the anisotropic case, the SGF-HOCBF baseline converges to a wrong boundary point near $(1.96,0.06)$ instead of $r^\star=(1.0425,1.6585)$; the final average distance to $r^\star$ is about $1.8431$, while the safety value remains nonnegative up to numerical precision. Thus, SGF-HOCBF preserves safety but may fail to preserve optimality when the optimal solution lies on the safety boundary.

\begin{figure}[t]
    \centering
    \includegraphics[width=0.95\linewidth]{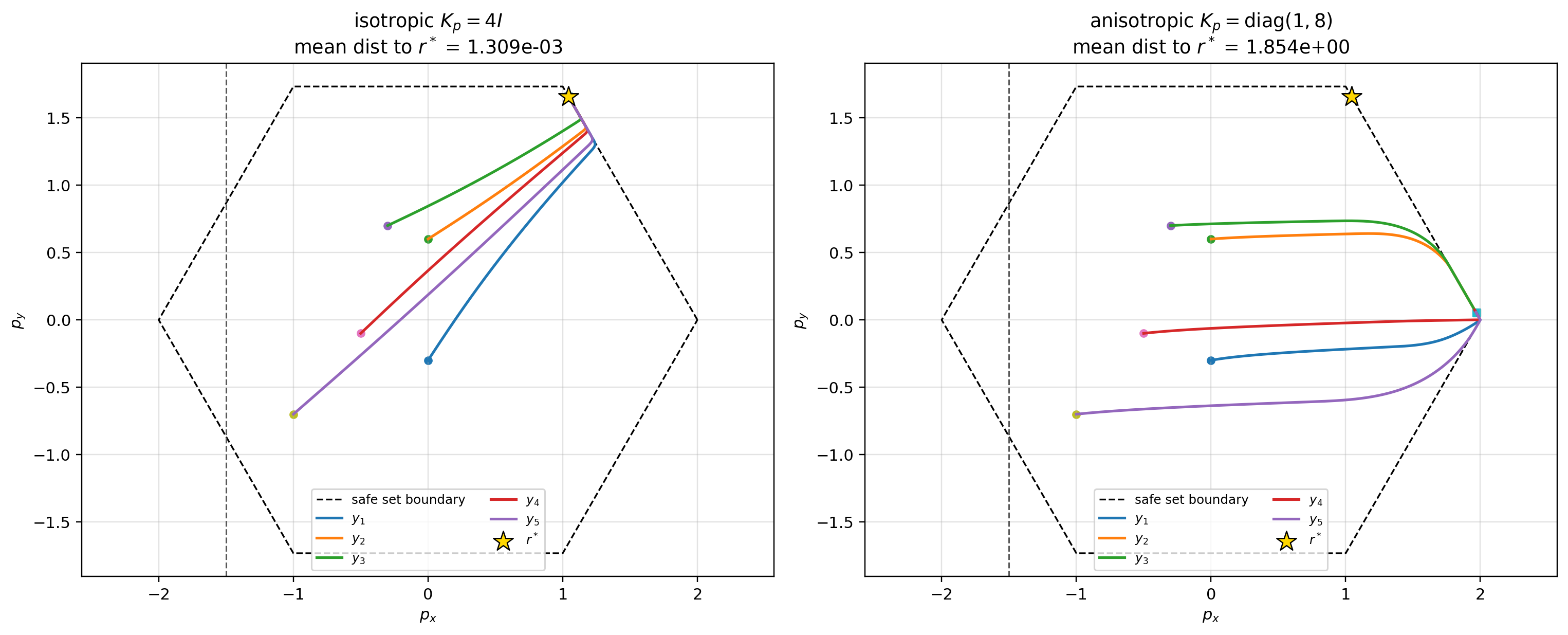}
    \caption{Trajectories of the SGF-HOCBF baseline under isotropic and anisotropic pre-stabilization.}
    \label{fig:sgf_hocbf_traj}
\end{figure}

Fig.~\ref{fig:pdfo_traj} shows the projected primal-dual feedback optimization baseline. This method can drive the outputs toward the constrained optimal solution, but the trajectories leave the safe workspace during transients. This is expected because the projection in \eqref{eq:sim_pdfo} is only used to maintain the nonnegativity of the dual variable $z_i$, and does not enforce forward invariance of the output safety set.

\begin{figure}[t]
    \centering
    \includegraphics[width=0.7\linewidth]{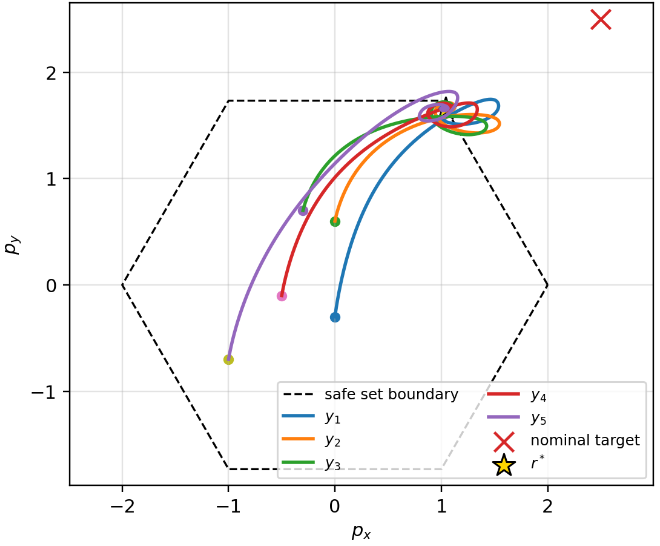}
    \caption{Trajectories of the projected primal-dual feedback optimization baseline. The method handles steady-state inequality constraints, but transient output safety is not guaranteed.}
    \label{fig:pdfo_traj}
\end{figure}

Figs.~\ref{fig:proposed_iso_traj} and \ref{fig:proposed_aniso_traj} show the trajectories of the proposed method under isotropic and anisotropic tracking gains, respectively. Unlike SGF-HOCBF, the proposed method converges to the true constrained optimal solution in both cases. The final average output distance to $r^\star$ is about $7.61\times10^{-11}$ in the isotropic case and $1.20\times10^{-10}$ in the anisotropic case.

\begin{figure}[t]
    \centering
    \subfloat[Isotropic tracking gains.]{%
        \includegraphics[width=0.48\linewidth]{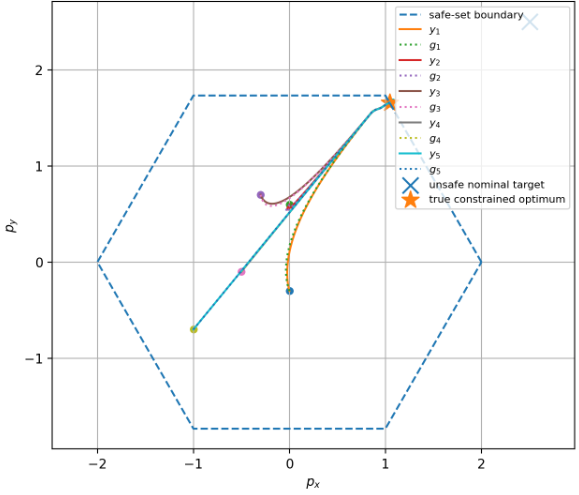}
        \label{fig:proposed_iso_traj}%
    }
    \hfill
    \subfloat[Anisotropic tracking gains.]{%
        \includegraphics[width=0.48\linewidth]{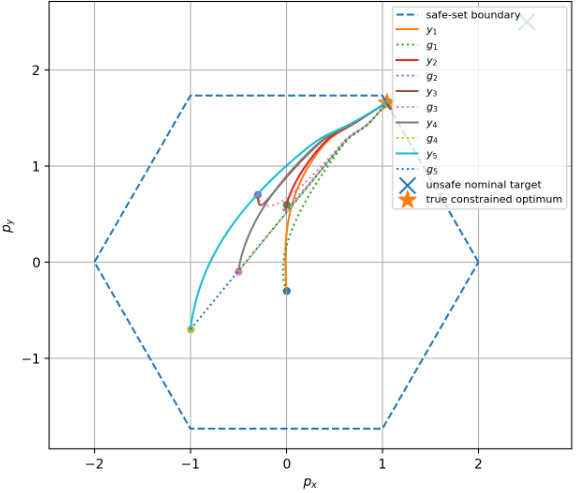}
        \label{fig:proposed_aniso_traj}%
    }
    \caption{Trajectories of the proposed DSM reference governor under isotropic and anisotropic tracking gains.}
    \label{fig:proposed_tracking_traj}
\end{figure}

The convergence processes are compared in Fig.~\ref{fig:convergence_compare}. The proposed method converges to the true constrained optimal solution in both tracking cases, whereas the SGF-HOCBF method with anisotropic pre-stabilization converges to a non-optimal boundary equilibrium.

\begin{figure}[t]
    \centering
    \subfloat[Proposed method.]{%
        \includegraphics[width=0.48\linewidth]{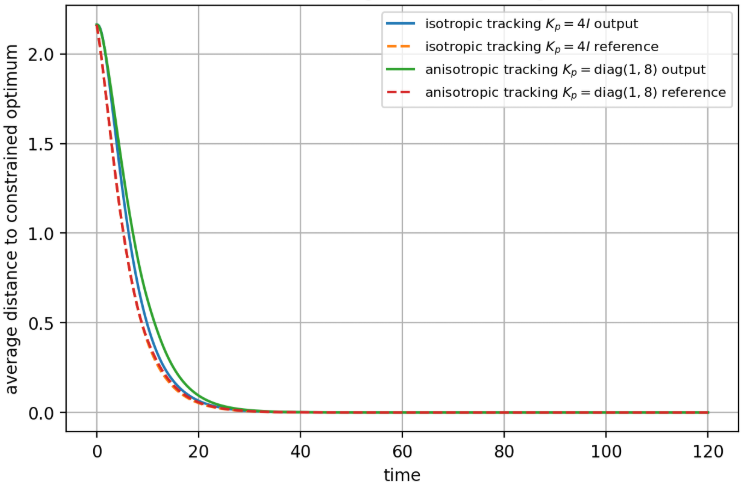}
        \label{fig:erg_convergence}%
    }
    \hfill
    \subfloat[SGF-HOCBF baseline.]{%
        \includegraphics[width=0.48\linewidth]{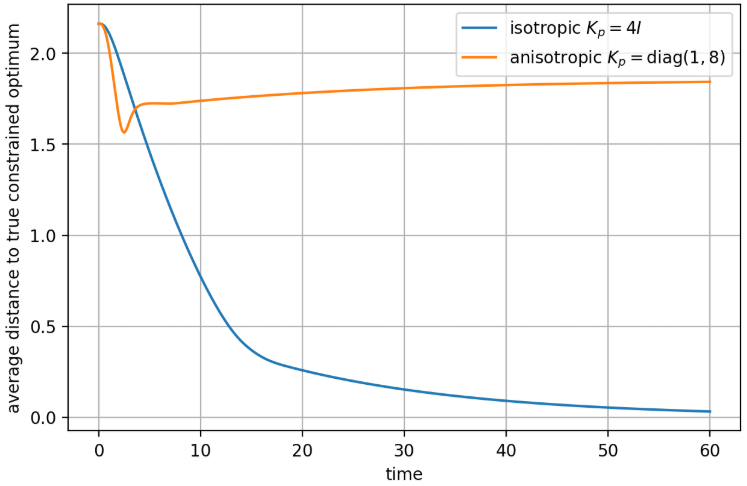}
        \label{fig:sgf_converge}%
    }
    \caption{Average distance to the constrained optimal solution $r^\star$.}
    \label{fig:convergence_compare}
\end{figure}

The safety diagnostics of the proposed method are shown in Figs.~\ref{fig:reference_safety} and \ref{fig:transient_safety}. The reference safety CBF keeps $\eta_\ell(g_i(t))\ge0$, while the DSM constraint keeps $M^S_{i\ell}(t)\ge0$, thereby certifying transient output safety. In the simulation, the minimum output safety value is about $7.61\times10^{-11}$ in the isotropic case and $8.68\times10^{-11}$ in the anisotropic case, and the minimum DSM value is zero up to numerical precision.

\begin{figure}[t]
    \centering
    \subfloat[Reference safety.]{%
        \includegraphics[width=0.48\linewidth]{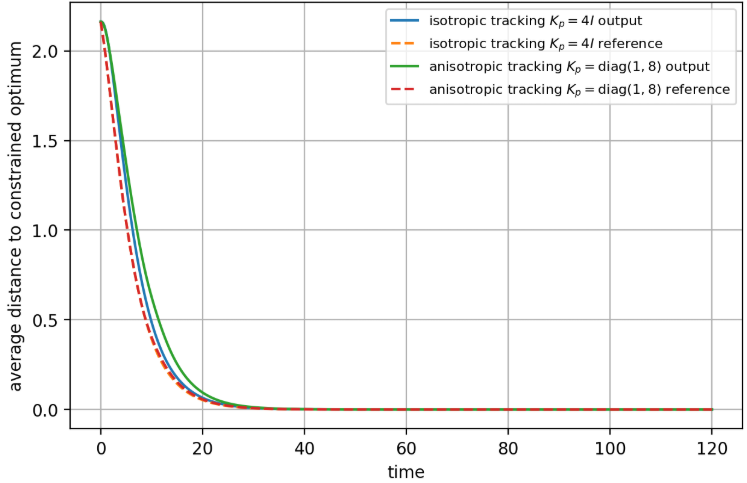}
        \label{fig:reference_safety}%
    }
    \hfill
    \subfloat[Transient safety DSM.]{%
        \includegraphics[width=0.48\linewidth]{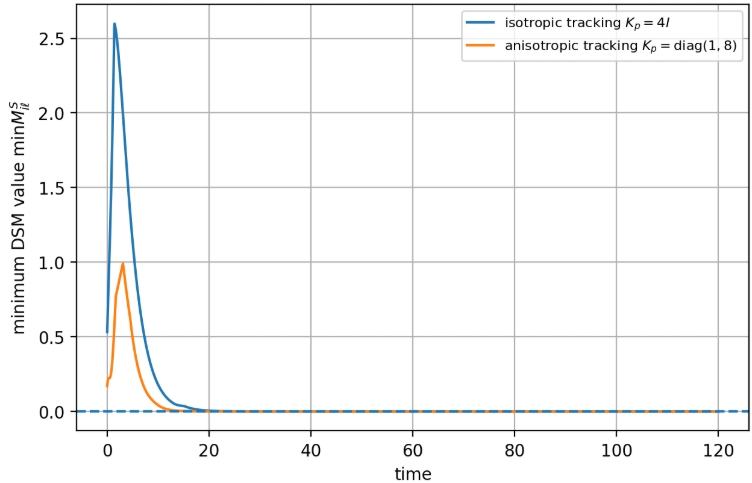}
        \label{fig:transient_safety}%
    }
    \caption{Safety diagnostics of the proposed DSM reference governor.}
    \label{fig:safety_diagnostics}
\end{figure}

We further test the proposed DSM reference governor under an additional
nonconvex obstacle constraint. A circular obstacle is placed at
$o=\operatorname{col}(0.5,1.0)$ with radius $r_o=0.25$, and the corresponding
safe set is described by
$\eta_{\rm obs}(r)=\|r-o\|^2-r_o^2\ge0$. The obstacle does not exclude the
desired hexagonal optimal solution $r^\star$, but it lies between the initial
references and $r^\star$, and therefore creates a possible spurious boundary
equilibrium for the CBF-filtered reference flow. In the baseline case, the
DSM reference governor is augmented with the obstacle reference CBF
$\nabla\eta_{\rm obs}(g_i)^\top\rho_i+\alpha_{\rm obs}\eta_{\rm obs}(g_i)\ge0$
and the corresponding obstacle DSM-CBF, with $\alpha_{\rm obs}=0.8$ and
$\alpha_{\rm Mobs}=5$. The isotropic tracking gains $K_p=4I$ and $K_d=4I$ are
used.

To test the proposed escape mechanism, we also implement adaptive tangential
objective shaping. The QP constraints are unchanged, and only the nominal
gradient direction is modified by replacing $-\nabla c_i(y_i)$ with
$-\nabla c_i(y_i)-\nabla_{g_i}\Phi_i(g_i,\vartheta_i)$. The shaping term uses
the tangent projection of the circular obstacle boundary, with $Q_i=I$, a
Gaussian bump radius $0.80$, $\bar\sigma_i=50$, $k_{\sigma i}=6$,
$k_{\xi i}=0.10$, $h_{i,\rm off}=0.30$, and alignment threshold $0.15$. The
shaping state evolves continuously according to
$\dot\sigma_i=-k_{\sigma i}(\sigma_i-\bar\sigma_i a_i)$ and
$\dot\xi_i=k_{\xi i}a_i(g_i-\xi_i)$, where $a_i$ is a smooth activation
function that becomes large near the obstacle boundary when
$\nabla c_i(y_i)$ is nearly aligned with $\nabla\eta_{\rm obs}(g_i)$.

\begin{figure}[t]
    \centering
    \includegraphics[width=0.9\linewidth]{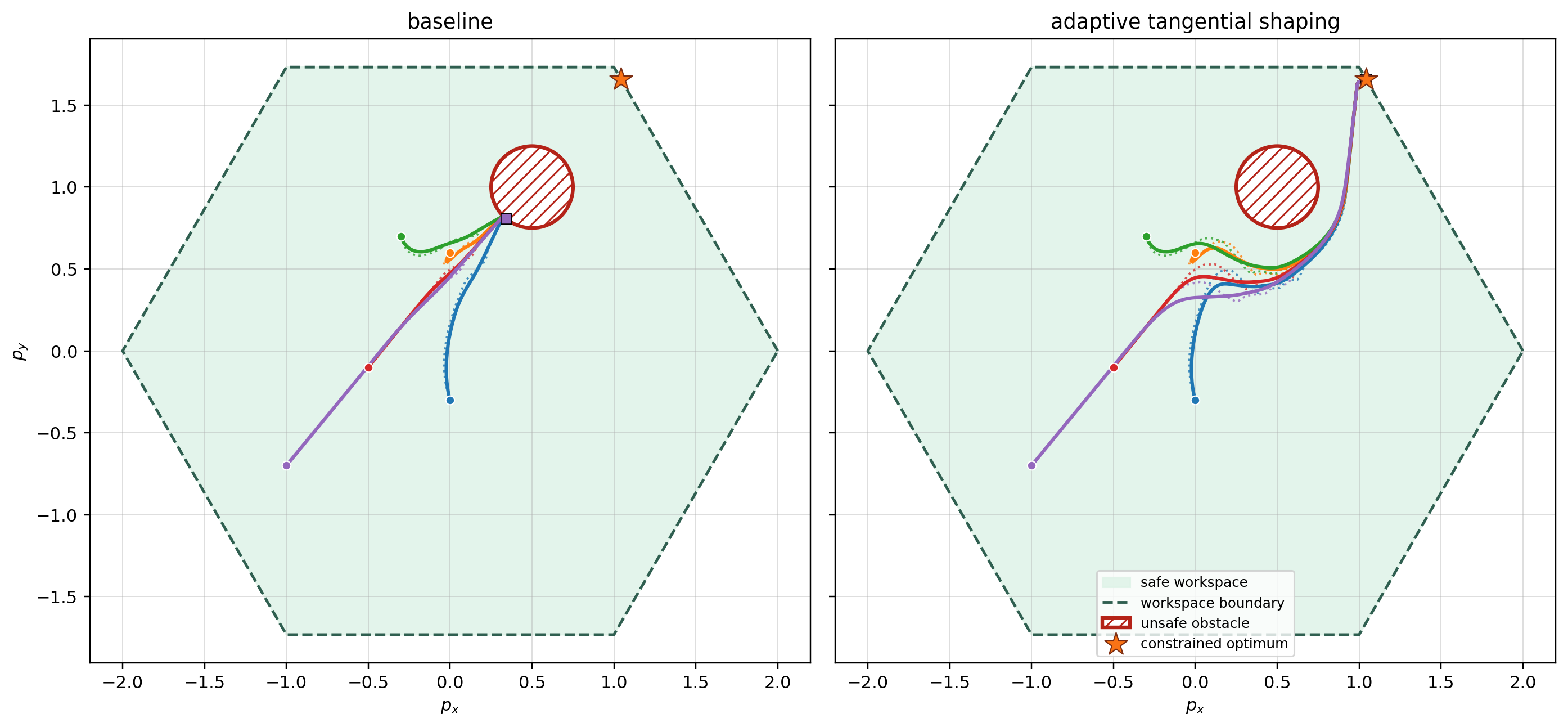}
    \caption{Trajectories with a nonconvex circular obstacle: DSM reference governor without shaping and with adaptive tangential objective shaping.}
    \label{fig:obstacle_tangential_shaping}
\end{figure}

Fig.~\ref{fig:obstacle_tangential_shaping} compares the trajectories. Without
shaping, the agents reach agreement but converge to a non-optimal point on the
obstacle boundary. The final mean output is approximately
$\operatorname{col}(0.3407,0.8052)$, the final average distance to
$r^\star$ is about $1.1048$, and the final consensus residual is about
$7.38\times10^{-14}$. Hence the failure is not caused by lack of agreement,
but by a stable nonconvex-induced boundary equilibrium. With adaptive
tangential shaping, the trajectories acquire a tangential component near the
obstacle, move around the boundary, and converge to the desired optimal solution. The
final mean output is approximately $\operatorname{col}(1.04245,1.65848)$, the
final average distance to $r^\star$ is about $2.15\times10^{-5}$, and the
final consensus residual is about $1.26\times10^{-13}$. The shaped run also
keeps a positive obstacle safety margin, with minimum output and reference
obstacle margins about $0.1028$ and $0.1098$, respectively. This experiment
illustrates that tangential shaping can turn an attracting nonconvex boundary
equilibrium into an escapable saddle-type region, while leaving the CBF and
DSM safety constraints unchanged.

Overall, SGF-HOCBF enforces safety but may converge to a spurious boundary
equilibrium, while projected primal-dual feedback optimization handles
steady-state inequalities but not transient safety. The proposed DSM governor
uses reference CBFs and DSM-CBFs to preserve transient safety and convergence
for convex workspace constraints. The nonconvex-obstacle case further shows
that stable boundary equilibria can occur, and that adaptive tangential
objective shaping provides a local escape mechanism without relaxing safety.

\section{Conclusion and Future Work}

This paper proposed a reference-governed distributed safe gradient-flow
framework for safe optimal output agreement. By separating output regulation
from distributed optimization, first-order reference CBFs enforce admissibility
and DSM-CBFs certify transient output safety. Under DSM-compatibility and
regularity conditions, the network-sparse QP preserves static optimality and
the coupled dynamics converge by a small-gain argument. Simulations verified
safe convergence, advantages over HOCBF-based feedback optimization, and the
ability of adaptive tangential objective shaping to escape nonconvex-induced
spurious equilibria. Future work will study optimal solutions located on nonconvex
constraint boundaries and extend the analysis to nonconvex distributed feedback optimization beyond output-agreement tasks.

\bibliographystyle{unsrtnat}
\bibliography{reference}

\end{document}